\documentclass[12pt,a4paper,final]{iopart}

\usepackage{iopams}  
\usepackage[dvips]{graphicx}
\usepackage[breaklinks=true,colorlinks=true,linkcolor=blue,urlcolor=blue,citecolor=blue]{hyperref}
\expandafter\let\csname equation*\endcsname\relax
\expandafter\let\csname endequation*\endcsname\relax
\usepackage{amsmath}
\usepackage{ulem}
\usepackage{cancel}
\usepackage{amsthm}
\usepackage{mathtools}
\usepackage{tikz}
\usepackage{graphicx}
\usepackage{mathrsfs}
\usepackage{multirow}
\graphicspath{ {./images/} }
\newtheorem{theorem}{Theorem}[section]
\newtheorem{corollary}{Corollary}[theorem]
\newtheorem{lemma}[theorem]{Lemma}
\newcommand{\vct}{\textbf}
\newcommand{\spc}{\qquad}
\newcommand{\npg}{\vspace{\baselineskip}}

\newcommand{\npgni}{\npg\noindent}
\newcommand{\tph}{^\frac{1}{2}}

\begin{document}
\title{3-Body Problems, Hidden Constants, Trojans and WIMPs (NQG II)}

\author{Richard Durran, Aubrey Truman}
\address{Department of Mathematics, Computational Foundry, Swansea University Bay Campus, Fabian Way, Swansea, SA1 8EN, UK}
\ead{a.truman@swansea.ac.uk}

\begin{abstract}

This work includes two new results - principally two new constants of motion for the linearised restricted 3-body problem and an important isosceles triangle generalisation of 
Lagrange's equilateral triangle solution of the restricted case leading to hidden constants for Hildans as well as Trojans. Both of these results are classical, but we also have included new results on Newtonian quantum gravity emanating from the asymptotics relevant for WIMPish particles, explaining the origin of systems like that of the Trojans. The latter result uses a generalisation of our semi-classical mechanics for Schrödinger equations involving vector as well as scalar potentials, presented here for the first time, thereby providing an acid test of our ideas in predicting the quantum curvature and torsion of WIMPish trajectories for our astronomical elliptic states. The combined effect is to give a new celestial mechanics for WIMPs in gravitational systems as well as new results for classical problems. As we shall explain, we believe these results could help to see how spiral galaxies evolve into elliptical ones, giving a simple fluid model in this connection. A simple classical consequence of our isosceles triangle result gives a Keplerian type $4^{\textrm{th}}$ Law for 3-body problems.
    
\end{abstract}

\section{Introduction}

\npgni The paper is in two parts. The first part is classical and is based on Newton's laws. The second part is quantum mechanical and rests upon Newtonian quantum gravity for WIMPs. 

\npgni \textbf{Part 1} (Sections 2 and 3)

\npgni Lagrange’s equilateral triangle solution to the 3-body problem (first published in 1772) was not thought to have any applications to astronomy until 1905, when it was realised that the Trojan asteroids, (60 degrees ahead or behind Jupiter on the same orbit, assumed to be a circle centred at O, the mass centre of the Sun(S) and Jupiter(J)) provided a working example. In the usual case the Sun and Jupiter move on circles, also centred at O, as required by the SJ 2-body problem, the asteroid being assumed to have negligible mass but still subject to the gravitational attractions of the Sun and Jupiter. For equilibrium the asteroid (A) is positioned at the point L, so that triangle SLJ is equilateral. In this circular case the side length of the triangle, $a_{0}$, is constant. Lagrange’s solution is more general e.g. it allows for the orbits of the 3 bodies to be similar ellipses with the same eccentricity, e, and common focus O. When e is not zero, the triangle side lengths are not constant but have to vary so as to maintain the initial side length ratios. So, if the original configuration of the 3 bodies is an equilateral triangle, it remains an equilateral triangle whose size changes as it rotates. For the restricted 3-body problem we generalise Lagrange's result to isosceles triangles - the most general case possible as we explain. To begin with, it is the motion of the asteroid A in a neighbourhood of L in the 2 cases $e=0$ and $e\sim0$ we investigate in this work, SLJ being an equilateral triangle.(See F.R. Moulton Ref.[19]).

\npgni The theorems of Bruns and Poincare on constants of the motion for the 3-body problem are presumed to preclude the existence of new constants not dependent on the so-called classical ones. However, they do not forbid there being new constants in restricted regions of phase-space e.g. for the above linearised or full restricted 3-body problem considered here. For the case of circular orbits we present 2 new constants of the motion in the rotating frame of the corresponding linearised SJ 2-body problem with rotating axes $OX$, $OY$, where $OX$ is parallel to $O\textrm{J}$. (See E. Whittaker chapter XIV Ref.[35] 
and W. Thirring Vol 1 Ref.[31]). We also reveal hidden constants of the motion for isosceles triangle configurations for the restricted 3-body problem, concomitant with our new isosceles triangle solution e.g. for Hildan asteroids.

\npgni It is well known that in the, $e=0$, circular orbit case, the Jacobi integral, $J$, is the constant Hamiltonian for the asteroid motion in the rotating frame. We show here that $J=f(D_{1},D_{2})$ where $D_{1}$ and $D_{2}$ are the new constants, $f$ a simple function. When $e$ is nonzero, we consider the asymptotic behaviour of the solution of our equations of motion as $e\sim0$, thereby finding 6 explicit constants of integration for the first order approximation to the solution. We believe this result generalises to arbitrarily high orders in powers of $e$, but the algebraic complications would require a lot of computation to determine any constants. These last results are achieved by realising that $J$ can be viewed as an electromagnetic Hamiltonian which we have used before in Burgers-Zeldovich models of the early history of galaxies and solar systems.(See Ref.[20]).

\npgni \textbf{Part 2} (Sections 4 and 5)

\npgni Unlike the first part of the paper where in line with classical theory we try to predict the future behaviour of the motion in the linearised restricted 3-body problem (e.g. the motion of the Trojan asteroids) by identifying two new constants of the motion, in the second part, in the context of Newtonian quantum gravity for our astronomical elliptic states, we try to predict the past. Namely we try to explain how celestial bodies such as Trojan asteroids could have condensed out of a cloud of WIMPish particles onto periodic elliptical orbits with force centre at the centre (not the focus) of the ellipse. Needless to say this inevitably involves the Schrödinger equation for the isotropic harmonic oscillator in 2 and 3-dimensions associated with a homogeneous cloud of gravitating WIMPs and very special initial conditions. (See Lena, Delande and Gay Ref.[17]).

\npgni To be specific we compute the Bohr correspondence limit of the analogue for the isotropic harmonic oscillator potential of Lena et al's atomic elliptic state for the Schrödinger equation. The leading term of Nelson's stochastic mechanics as $\epsilon^2=\frac{\hbar}{m}\sim0$, $\hbar$ being Planck's constant, $m$ the mass of the diffusing WIMPish particle, defines our semi-classical mechanics. That this captures the correspondence limit of Nelson's stochastic mechanics, more generally, even in the presence of singularities, can be deduced from Freidlin and Wentzell. (See Ref.[13]). In our case this does indeed give the desired convergence to motion on an ellipse as required for the Trojans.

\npgni We contend that this is a sensible approach to understanding better the formation of celestial systems in spite of the caution in Landau and Lifshitz; "To describe the state of a macroscopic body by a wave function at all is impracticable, since the available data concerning the state of such a body are far short of the complete set of data necessary to establish its wave function." (See Ref.[16]). We believe very strongly that it is highly profitable to compute the Schrödinger wave functions for different component parts of protoplanetary nebulae. Indeed this is our main strategy. 

\npgni An added complication here is that the linearised restricted 3-body problem involves rotating coordinates and so involves a vector as well as a scalar potential as we shall see. Here we generalise our semi-classical mechanics in line with the Schrödinger equation and 
calculate the quantum corrections to curvature and torsion of the trajectories of the WIMPish particles condensing into the Trojan asteroids. This enables us to test our ideas on other embryonic restricted 3-body problems given we have explicit information on the quantum corrections in this case. There are important ramifications of these ideas for the evolution of galaxies as we shall explain by developing a Burgers-Zeldovich fluid model with vorticity and viscosity incorporating our ideas and extending the scope of applications considerably.

\npgni \textbf{Commentary}

\npgni A few final remarks about our treatment of the above problems, firstly, unlike most previous authors, we have made extensive use of Hamilton-Lenz-Runge vectors in considering the Lagrange results. Our predilection in this connection owes everything to Pauli's treatment of the quantum Kepler/Coulomb problem which was vitally important in Lena et al's derivation of their atomic elliptic state whose semi-classical limit underlies the present work and its connection to constants for WIMPs. We should add that, if at time $t$, $\overset{\rightharpoonup}{\textrm{LA}}=(X(t),Y(t))$ in rotated coordinates, then both $X$ and $Y$ are almost periodic functions, $f:\mathbb{R}_{+}\rightarrow\mathbb{R}$. For any such function $f$, important roles are played by

\npgni $$\textrm{P}_{r}(f)=\lim_{T\nearrow\infty}T^{-1}\int^T_0\textrm{e}^{-irs}f(s)ds \;\;\;\textrm{and}\;\;\; \sigma_{\textrm{B}}(f)=\{r\in\mathbb{R}:\textrm{P}_{r}(f)\ne0\}.$$

\npgni The B stands for Haage Bohr, the brother of Niels who developed the notion in quantum mechanics of the spectrum of a quantum observable realised as a linear operator e.g. $\hat{H}$ the quantised version of our electromagnetic Hamiltonian herein, its spectrum is given by,

$$\sigma_{\textrm{B}}[\hat{H}]=\{\lambda\in\mathbb{R}:(\lambda I-\hat{H})\;\textrm{is not invertible for}\;I\;\textrm{the identity}\}.$$

\npgni For periodic orbits there are strong links between $\sigma_{\textrm{B}}(X)$, $\sigma_{\textrm{B}}(Y)$ and $\sigma_{\textrm{B}}[\hat{H}]$ for Trojan asteroids as we shall see.

\npgni Our one true inspiration is in drawing attention to the result that triangle SAJ is isosceles i.e. $|\textrm{SA}|=|\textrm{JA}|$ if and only if the angular momentum of A about $O$, the mass centre of S and J, is a constant in time. This result and the new constants seem to have escaped the attention of previous authors. This is all the more amazing when one realises that the new constants have ramifications for the Foucault pendulum as well as other classical problems and, even though the isosceles triangle results emanate from the topical Trojan orbits, no one has investigated the possibilities of the corresponding orbits for Hildans. Even more surprising are the isosceles triangle results on Kepler's $3^{\textrm{rd}}$ Law detailed in the Appendix. To verify our results further, we clearly will need more data.

\section{Hidden Constants for Circular Orbits}

\subsection{Background for Linearisation of Circular Orbital Case of 3 Bodies}

\npgni The Hamiltonian, $H$, for the motion of the asteroid A in the centre of mass frame for the two bodies S and J, with coordinates $Ox$, $Oy$ forming our inertial frame, is

$$H=2^{-1}(p_{x}^2+p_{y}^2)+V,\;\;\;V=-\frac{\mu_{1}}{|\textrm{SA}|}-\frac{\mu_{2}}{|\textrm{JA}|},$$

\npgni $\overset{\rightharpoonup}{O\textrm{A}}=(x,y)$ with conjugate momenta $(p_{x},p_{y})$, $\mu_{1}$ and $\mu_{2}$ being the gravitational masses of the Sun and Jupiter, respectively.

\npgni Denoting by $\omega$ the angular velocity of Jupiter on its circular orbit around $O$, the classical contact transformation from $(x,y)$ coordinates to $(\tilde x,\tilde y)$ coordinates in the rotating frame, is defined by:-

$$x=\frac{\partial W}{\partial p_{x}},\;\;\;y=\frac{\partial W}{\partial p_{y}},\;\;\;\tilde p_{x}=\frac{\partial W}{\partial \tilde x},\;\;\;\tilde p_{y}=\frac{\partial W}{\partial \tilde y},$$

\npgni for

$$W=p_{x}(\tilde x\cos\omega t-\tilde y\sin\omega t)+p_{y}(\tilde x\sin\omega t+\tilde y\cos\omega t),$$

\npgni giving the new Hamiltonian in the rotating frame

$$K=H-\frac{\partial W}{\partial t}=\frac{1}{2}(p_{x}^2+p_{y}^2)+\omega(\tilde y\tilde p_{x}-\tilde x\tilde p_{y})+\tilde V,$$\vspace{-10mm}

\npgni i.e.\vspace{-5mm}

$$K=2^{-1}(\tilde {\vct{p}}-\tilde {\vct{A}})^{2}+\tilde V-\frac{\tilde {\vct{A}}^2}{2},$$

\npgni where $\tilde A=\omega(-\tilde y,\tilde x)$, with a similar interpretation for $\tilde V$.

\npgni Since $\dfrac{\partial K}{\partial t}=0$, the new Hamiltonian is constant in time and is the celebrated Jacobi\vspace{3mm}
integral, $J$, in the rotating coordinates $(\tilde x,\tilde y)$,

$$x=\tilde x\cos\omega t-\tilde y\sin\omega t,\;\;\;y=\tilde x\sin\omega t+\tilde y\cos\omega t.$$

\npgni The resulting Hamilton equations are:-

$$\ddot{\tilde x}-2\omega\dot{\tilde y}-\omega^{2}\tilde x=-\frac{\mu_{1}(\tilde x+r_{1})}{((\tilde x+r_{1})^{2}+{\tilde y}^{2})^{3/2}}-\frac{\mu_{2}(\tilde x-r_{2})}{((\tilde x-r_{2})^{2}+
{\tilde y}^{2})^{3/2}},$$

$$\ddot{\tilde y}+2\omega\dot{\tilde x}-\omega^{2}\tilde y=-\frac{\mu_{1}\tilde y}{((\tilde x+r_{1})^{2}+{\tilde y}^{2})^{3/2}}-\frac{\mu_{2}\tilde y}{((\tilde x-r_{2})^{2}+{\tilde y}^{2})^{3/2}},$$

\npgni where $r_{1}=|O\textrm{S}|$, $r_{2}=|O\textrm{J}|$, $(r_{1}+r_{2})=a_{0}$, is the triangle side length and $\omega^2=\dfrac{\mu_{1}+\mu_{2}}{a_{0}^3}$. These equations admit Lagrange's equilateral triangle solution,

$$(\tilde x,\tilde y)=(c,d)=\left(\frac{r_{2}-r_{1}}{2},\pm \frac{\sqrt{3}}{2}(r_{1}+r_{2})\right),$$

\npgni $\overset{\rightharpoonup}{O\textrm{L}}=(c,d)$ is the position vector of $\textrm{L}=\mathscr{L}_{4,5}$, the Lagrange equilibrium points.

\subsection{Equations of Motion in Rotated Coordinates}\vspace{-5mm}

\npgni Working in rotating coordinates and writing, $(\tilde x,\tilde y)=(c+\delta(t),d+\epsilon(t))$, where $\overset{\rightharpoonup}{O\textrm{L}}=(c,d)$, the Lagrange equilibrium point, $\mathscr{L}_{4,5}$, gives the linearised equations:-

$$\ddot{\delta}-2\omega \dot{\epsilon}-\frac{3}{4}{\omega}^2\delta-{\Omega}^2\epsilon=0,$$

$$\ddot{\epsilon}+2\omega \dot{\delta}-\frac{9}{4}{\omega}^2\epsilon-{\Omega}^2\delta=0,$$\vspace{-8mm}

\npgni where ${\Omega}^2=\dfrac{3\sqrt{3}(\mu_{1}-\mu_{2})}{4(r_{1}+r_{2})^3}$, $\mu_{1}\ge\mu_{2}$,\vspace{3mm}

$$\tilde C=\frac{1}{2}({\dot{\delta}^2}+{\dot{\epsilon}^2})-\frac{3}{8}{\omega}^2({\delta}^2+3{\epsilon}^2)-\Omega^2\epsilon\delta,$$

\npgni being the corresponding Jacobi integral. This suggests we make a fixed rotation $(\epsilon,\delta)\rightarrow(X,Y)$,

$$\delta=X\cos \gamma - Y\sin \gamma \spc;\spc \epsilon=X\sin \gamma + Y\cos \gamma,$$\vspace{-8mm}

\npgni where\vspace{-5mm}

$$\tan 2\gamma=-\dfrac{\sqrt{3}(\mu_{1}-\mu_{2})}{\mu_{1}+\mu_{2}},$$

\npgni and

$$C_{0}=\frac{1}{2}({\dot{X}^2}+{\dot{Y}^2})-\omega_{X}^2X^2-\omega_{Y}^2Y^2,\;\;\dot{X}=p_{X}-A_{X},\;\dot{Y}=p_{Y}-A_{Y}.$$\vspace{-8mm}

\npgni Here

$$\omega_{X}^2=\left(\dfrac{3}{4}-\left(\dfrac{3}{8}+\dfrac{2\Omega^4}{3\omega^4}\right)\cos 2\gamma\right)\omega^2,$$

$$\omega_{Y}^2=\left(\dfrac{3}{4}+\left(\dfrac{3}{8}+\dfrac{2\Omega^4}{3\omega^4}\right)\cos 2\gamma\right)\omega^2,$$

\npgni with $\omega_{X}^2+\omega_{Y}^2=\dfrac{3}{2}\omega^2$ and $\omega_{X}^2\omega_{Y}^2=\dfrac{27\mu_{1}\mu_{2}}{16(\mu_{1}+\mu_{2})^2}\omega^4$.

\npgni Moreover $X$ and $Y$ satisfy

$$\ddot{X}-2\omega\dot{Y}-2\omega_{X}^2X=0\spc;\spc\ddot{Y}+2\omega \dot{X}-2\omega_{Y}^2Y=0.$$

\npgni Assuming solutions of the form $X=A\textrm{e}^{\lambda t}$ and $Y=B\textrm{e}^{\lambda t}$ leads to

$$\lambda^4+\omega^2\lambda^2+4\omega_{X}^2\omega_{Y}^2=0,$$\vspace{-8mm}

\npgni with general solution of the form

$$X=2\omega {\sum_{\textrm{i}=1}^{4}}\lambda_\textrm{i}C_\textrm{i}\textrm{e}^{\lambda_\textrm{i}t}\spc;\spc Y={\sum_{\textrm{i}=1}^{4}}(\lambda_\textrm{i}^2-2\omega_{X}^2)C_\textrm{i}\textrm{e}^{\lambda_\textrm{i}t}.$$

\npgni Suppose further that all the roots are imaginary i.e. $(\lambda_{1},\lambda_{2},\lambda_{3},\lambda_{4})=(\alpha i, -\alpha i, \beta i, -\beta i)$, with $\alpha\ne\beta$, $i=\sqrt{-1}$.

\begin{theorem}(Hidden Constants for the Case $e=0$)

\npgni For the system

$$\ddot{X}-2\omega \dot{Y}-2\omega_{X}^2X=0\spc;\spc\ddot{Y}+2\omega \dot{X}-2\omega_{Y}^2Y=0,$$

\npgni for constant $\omega$ with $\omega_{X}^2+\omega_{Y}^2=\dfrac{3}{2}\omega^2$ and $\omega_{X}^2\omega_{Y}^2=\dfrac{27\mu_{1}\mu_{2}}{16(\mu_{1}+\mu_{2})^2}\omega^4$, in addition to the Jacobi constant,

$$J=\frac{1}{2}({\dot{X}^2}+{\dot{Y}^2})-\omega_{X}^2X^2-\omega_{Y}^2Y^2,\;\;\dot{X}=p_{X}-A_{X},\;\dot{Y}=p_{Y}-A_{Y},$$

\npgni there are two further constants of the motion:-

$$D_{1}=\alpha^2\{(\beta^2+2\omega_{X}^2)\dot{X}-2\omega\beta^2Y\}^2+4\omega_{X}^4\{2\omega\dot{Y}+(\beta^2+2\omega_{X}^2)X\}^2,$$

$$D_{2}=\beta^2\{(\alpha^2+2\omega_{X}^2)\dot{X}-2\omega\alpha^2Y\}^2+4\omega_{X}^4\{2\omega\dot{Y}+(\alpha^2+2\omega_{X}^2)X\}^2,$$

\npgni $\alpha$ and $\beta$ being the real roots of $t^4-\omega^2t^2+4\omega_{X}^2\omega_{Y}^2=0$, provided that $\dfrac{\mu_{1}\mu_{2}}{(\mu_{1}+\mu_{2})^2}<\dfrac{1}{27}$.

\npgni Moreover, setting $f(u)=u^{-2}((\omega^2+3\omega_{X}^2)u^4-\omega_{X}^2(13\omega^2-8\omega_{X}^2)u^2-4\omega_{X}^4\omega_{Y}^2)$ yields

$$32\omega_{X}^4\omega^2(\alpha^2-\beta^2)^2J=f(\alpha)D_{1}+f(\beta)D_{2}.$$

\end{theorem}

\begin{proof}
For the system defined above the solutions for $X$ and $Y$ together with their time derivatives can be expressed in matrix form:

$$\begin{pmatrix}{X}\\{\dot X}\\{Y}\\{\dot Y}\end{pmatrix}=\begin{pmatrix}2\omega\lambda_{1} & 2\omega\lambda_{2} & 2\omega\lambda_{3} & 2\omega\lambda_{4}\\2\omega\lambda_{1}^2 & 2\omega\lambda_{2}^2 & 2\omega\lambda_{3}^2 & 2\omega\lambda_{4}^2\\\lambda_{1}^2-2\omega_{X}^2 & \lambda_{2}^2-2\omega_{X}^2 & \lambda_{3}^2-2\omega_{X}^2 & \lambda_{4}^2-2\omega_{X}^2\\ \lambda_{1}(\lambda_{1}^2-2\omega_{X}^2) & \lambda_{2}(\lambda_{2}^2-2\omega_{X}^2) & \lambda_{3}(\lambda_{3}^2-2\omega_{X}^2) & \lambda_{4}(\lambda_{4}^2-2\omega_{X}^2) \end{pmatrix}\begin{pmatrix}{C_{1}\textrm{e}^{\lambda_{1}t}}\\{C_{2}\textrm{e}^{\lambda_{2}t}}\\{C_{3}\textrm{e}^{\lambda_{3}t}}\\{C_{4}\textrm{e}^{\lambda_{4}t}}\end{pmatrix}.$$

\npgni Inverting this matrix equation gives expressions for $\textrm{e}^{\lambda_{\textrm{i}}t},\;\textrm{i}=1,2,3,4$ in terms of $X,\dot{X},Y$\vspace{3mm}
\npgni and $\dot{Y}$. When $\dfrac{\mu_{1}\mu_{2}}{(\mu_{1}+\mu_{2})^2}<\dfrac{1}{27}$, $\lambda_{1}=\alpha i$, $\lambda_{2}=-\alpha i$, $\lambda_{3}=\beta i$ and $\lambda_{4}=-\beta i$, where $\alpha$\vspace{-3mm}
\npgni and $\beta$ are the real roots of $t^4-\omega^2t^2+4\omega_{X}^2\omega_{Y}^2=0$. Assuming that $\omega_{X}(\alpha^2-\beta^2)\ne0$,\vspace {-3mm}
\npgni constants of the motion, $D_{1}$ and $D_{2}$ follow from the fact that $\textrm{e}^{i\alpha t}\textrm{e}^{-i\alpha t}=1$ and similarly for $\beta$. Furthermore, 
calculating the Jacobi integral directly from the solutions for $X$ and $Y$ yields its connection with $D_{1}$ and $D_{2}$.

\end{proof}\vspace{-5mm}

\subsection{Solutions of Equations of Motion X(t) and Y(t)}\vspace{-5mm}

\npgni For c.c. complex conjugate

$$X(t)=(2\omega\alpha iC_{1}\textrm{e}^{i\alpha t}+\textrm{c.c.})+(2\omega\beta iC_{3}\textrm{e}^{i\beta t}+\textrm{c.c.})$$\vspace{-12mm}

\npgni and

$$Y(t)=(-(\alpha^2+2\omega_{X}^2)C_{1}\textrm{e}^{i\alpha t}+\textrm{c.c.})+(-(\beta^2+2\omega_{X}^2)C_{3}\textrm{e}^{i\beta t}+\textrm{c.c.})$$\vspace{-10mm}

\npgni $\textrm{for}\;t\ge0$.

$$C_{1}\textrm{e}^{i\alpha t}=\frac{i}{8\omega_{X}^2\omega\alpha(\alpha^2-\beta^2)}\{i\alpha(\beta^2+2\omega_{X}^2)\dot X+2\omega i\beta^2\alpha Y+2\omega_{X}^2(2\omega\dot Y+(\beta^2+2\omega_{X}^2)X\}\rceil_{t}$$\vspace{-10mm}

\npgni and\vspace{-5mm}

$$C_{3}\textrm{e}^{i\beta t}=\frac{-i}{8\omega_{X}^2\omega\beta(\alpha^2-\beta^2)}\{i\beta(\alpha^2+2\omega_{X}^2)\dot X+2\omega i\alpha^2\beta Y+2\omega_{X}^2(2\omega\dot Y+(\alpha^2+2\omega_{X}^2)X\}\rceil_{t}$$

\npgni Evidently $|C_{1}|$ and $|C_{3}|$ are constants of the motion since the r.h.s in the last two equations have to be evaluated at time $t$. These are in essence the two new constants of the motion $D_{1}$ and $D_{2}$. If $C_{1}=|C_{1}|\textrm{e}^{i\phi_{1}}$ and $C_{3}=|C_{3}|\textrm{e}^{i\phi_{3}}$, $\phi_{1}$ and $\phi_{3}$ can be read off the above identities evaluated at $t=0$ giving

$$X(t)=-4\omega\alpha|C_{1}|\sin (\alpha t+\phi_{1})-4\omega\beta|C_{3}|\sin (\beta t+\phi_{3}),$$

$$Y(t)=-(\alpha^2+2\omega_{X}^2)|C_{1}|\cos (\alpha t+\phi_{1})-(\beta^2+2\omega_{X}^2)|C_{3}|\cos (\beta t+\phi_{3})$$

\npgni and as expected $X$ and $Y$ are almost periodic if there is no rational relationship between $\alpha$ and $\beta$. Needless to say one can read off the Jacobi integral constant from the above formulae and the initial conditions for periodicity are $|C_{1}|$ or $|C_{3}|$ has to be zero.

\section{More Constants for Equilateral Triangle Case and Generalisations to Isosceles Triangles}

\subsection{Equilateral Triangle and Isosceles Triangle Results}\vspace{-5mm}

\npgni In our laboratory inertial frame if the position vector of the centre of mass of three particles is $\vct R$, $\ddot{\vct R}=0$ and writing ${\vct r}_{\textrm{i}}={\vct R}_{\textrm{i}}-\vct R$, $\textrm{i}=1,2,3$, where $\textrm{i}=3$ corresponds to the asteroid position vector,

$$\ddot{\vct r}_{3}=-\mu_{1}\frac{({\vct r}_{3}-{\vct r}_{1})}{|{\vct r}_{3}-{\vct r}_{1}|^3}-\mu_{2}\frac{({\vct r}_{3}-{\vct r}_{2})}{|{\vct r}_{3}-{\vct r}_{2}|^3},\;\;\;\textrm{where}\;\;\sum_{\textrm{i}=1}^{3}\mu_{\textrm{i}}\vct r_{\textrm{i}}=\boldsymbol{0}.$$

\npgni So, if we look for a solution in which

$$|{\vct r}_{1}-{\vct r}_{2}|=|{\vct r}_{2}-{\vct r}_{3}|=|{\vct r}_{3}-{\vct r}_{1}|=a_{0}(t)$$\vspace{-10mm}

\npgni we must have

$$\ddot{\vct r}_{\textrm{i}}=-\frac{(\mu_{1}+\mu_{2}+\mu_{3})}{a_{0}^3(t)}{\vct r}_{\textrm{i}},\;\;\textrm{i}=1,2,3,\;\;\;\vct{h}_{\textrm{i}}=\vct{r}_{\textrm{i}}\times\dot{\vct{r}}_{\textrm{i}},\;\textrm{constant}\;\textrm{for}\;\textrm{i}=1,2,3.$$

\npgni So, setting $r_{\textrm{i}}=|{\vct r}_{\textrm{i}}(t)|$ resolving forces radially

$$\ddot{r_{\textrm{i}}}-{\dot\theta_{\textrm{i}}^2}r_{\textrm{i}}=-(\mu_{1}+\mu_{2}+\mu_{3})a_{0}^{-3}(t)r_{\textrm{i}},\;\;\;\textrm{i}=1,2,3,$$\vspace{-10mm}

\npgni i.e.\vspace{-5mm}

$$\frac{1}{r_{\textrm{i}}}\frac{d}{dr_{\textrm{i}}}\left(\frac{{\dot r_{\textrm{i}}}^2}{2}+\frac{h_{\textrm{i}}^2}{2r_{\textrm{i}}^2}\right)=-\frac{(\mu_{1}+\mu_{2}+\mu_{3})}{a_{0}^3(t)},\;\;\;\textrm{i}=1,2,3,$$

\npgni for $h_{\textrm{i}}=r_{\textrm{i}}^2\dot\theta_{\textrm{i}}$. Energy conservation for $E_{\textrm{i}}$, which we justify a posteriori, gives

$$\frac{1}{r_{\textrm{i}}}\frac{d}{dr_{\textrm{i}}}\left(\frac{{\tilde\mu_{\textrm{i}}}}{r_{\textrm{i}}}\right)=-\frac{(\mu_{1}+\mu_{2}+\mu_{3})}{a_{0}^3(t)},\;\;\;\textrm{i}=1,2,3.$$

\npgni So as long as we can assume we can choose constant $\tilde\mu_{\textrm{i}}$ and negative energies,

\npgni $E_{\textrm{i}}=2^{-1}\left(\dot r_{\textrm{i}}^2+\dfrac{h_{\textrm{i}}^2}{r_{\textrm{i}}^2}\right)-\dfrac{\tilde\mu_{\textrm{i}}}{r_{\textrm{i}}}$, such that

$$\ddot{\vct r}_{\textrm{i}}=-\frac{\tilde\mu_{\textrm{i}}}{r_{\textrm{i}}^{3}}{\vct r}_{\textrm{i}},\;\;\;\textrm{i}=1,2,3,$$\vspace{-8mm}

\npgni we will get Keplerian motion on an ellipse corresponding to this inverse square law. So assuming motion is rotation plus scaling we need

$$\tilde{\mu_{\textrm{i}}}=(\mu_{1}+\mu_{2}+\mu_{3})\frac{r_{\textrm{i}}^3(0)}{a_{0}^3(0)},\;\;\;\textrm{i}=1,2,3.$$\vspace{-5mm}

\npgni For this choice, individual energy and angular momentum are conserved, but we need to assume $\dot\theta_{1}=\dot\theta_{2}=\dot\theta_{3}$, so the triangle whilst it rotates remains equilateral and is subject to a simple scale factor. In any case we know from Newton's results that for the $\textrm{i}^{\textrm{th}}$ semi-latus rectum $l_{\textrm{i}}$

$$l_{\textrm{i}}=\frac{h_{\textrm{i}}^2}{\tilde{\mu_{\textrm{i}}}}=\frac{h_{\textrm{i}}^2(0)}{\tilde{\mu_{\textrm{i}}}}=\frac{r_{\textrm{i}}^2(0)\dot\theta_{\textrm{i}}(0)}{\tilde{\mu_{\textrm{i}}}}$$

\npgni and for the eccentricity of the orbit of the $\textrm{i}^{\textrm{th}}$ particle

$$e_{\textrm{i}}=\sqrt{1+\frac{2h_{\textrm{i}}^2E_{\textrm{i}}}{\tilde\mu_{\textrm{i}}^2}},\;\;\;\textrm{i}=1,2,3,$$\vspace{-10mm}

\npgni as long as\vspace{-5mm}

$$\ddot{\vct r}_{\textrm{i}}=-\frac{\tilde\mu_{\textrm{i}}}{r_{\textrm{i}}^{3}}{\vct r}_{\textrm{i}},\;\;\;\textrm{i}=1,2,3,$$

\npgni $\vct{A}_{\textrm{i}}=\dot{\vct{r}}_{\textrm{i}}\times\vct{h}_{\textrm{i}}+{\tilde\mu}_{\textrm{i}}r_{\textrm{i}}^{-1}\vct{r}_{\textrm{i}}$, constant for $\textrm{i}=1,2,3$, $\vct{A}_{\textrm{i}}$ being the Hamilton-Lenz-Runge vectors.

\npgni So we ask: "what conditions guarantee this?" Evidently the scale factor $s(t)=\dfrac{a_{0}(t)}{a_{0}(0)}$, so

$$s(t)=\frac{r_{\textrm{i}}(t)}{r_{\textrm{i}}(0)}=\frac{1+e_{\textrm{i}}\cos\theta_{\textrm{i}}(0)}{1+e_{\textrm{i}}\cos\theta_{\textrm{i}}(t)},\;\;\;\textrm{i}=1,2,3,$$

\npgni and we see it is necessary that $e_{1}=e_{2}=e_{3}=e$ and, if $e\ne0$, $\theta_{1}(0)=\theta_{2}(0)=\theta_{3}(0)=\theta_{0}$ and then $\dot\theta_{1}=\dot\theta_{2}=\dot\theta_{3}$ for all $t\ge0$. But

$$\dot\theta_{\textrm{i}}(t)=\frac{h_{\textrm{i}}}{s^2(t)r_{\textrm{i}}^2(0)},\;\;\;\textrm{i}=1,2,3,$$

\npgni so it is also necessary that $\dfrac{h_{\textrm{i}}}{r_{\textrm{i}}^2(0)}=c$, for 
$\textrm{i}=1,2,3,$ $c$ being the same constant i.e.\vspace{3mm} $h_{1}(0):h_{2}(0):h_{3}(0)=r_{1}^2(0):r_{2}^2(0):r_{3}^2(0)$, the angles between the Hamilton-Lenz-Runge vectors and ${\vct r}_{\textrm{i}}(0)$ having the same common value $\theta_{0}$ and length, $A_{\textrm{i}}=|{\vct A}_{\textrm{i}}|$, where

$$A_{\textrm{i}}=\tilde\mu_{\textrm{i}}e,\;\;\;\textrm{i}=1,2,3,$$

\npgni $e$, with $0<e<1$, can be any value if the above conditions are satisfied. This will be so if we set $\theta_{\textrm{i}}(t)=\theta_{0}(t)$, the true anomaly for the orbit of Jupiter denoted as $\theta_{0}(t)$. Therefore,

$$s(t)=\frac{1+e\cos\theta_{0}(0)}{1+e\cos\theta_{0}(t)},\;\;\;\frac{h_{1}}{r_{1}^2(0)}=\frac{h_{2}}{r_{2}^2(0)}=\frac{h_{3}}{r_{3}^2(0)},\;\;\;\;\;\;\;\;\;\;\;\;\;\;(\ast)$$

\npgni the $\dot r_{\textrm{i}}(0)$ having to be arranged so that\vspace{3mm}

$$e=\sqrt{1+\frac{2h_{\textrm{i}}^2E_{\textrm{i}}}{\tilde\mu_{\textrm{i}}^2}},\;\;\;E_{\textrm{i}}=\frac{1}{2}\left(\dot r_{\textrm{i}}^2(0)+\frac{h_{\textrm{i}}^2}{r_{\textrm{i}}^2(0)}\right)-\frac{\tilde\mu_{\textrm{i}}}{r_{\textrm{i}}(0)},\;\;\;\;\;(\ast\ast)$$

\npgni for $\textrm{i}=1,2,3$. It is simple to check that in this case the transverse equations of motion are also satisfied because $h_{\textrm{i}}$ has the correct value. Lagrange's most general result can be proved in a similar fashion.

\npgni The following result is relevant here. We use the above nomenclature in which potentially the Sun is particle 1 with position $\vct{r}_{1}$ in the centre of mass frame, Jupiter is particle 2 and an asteroid of negligible mass particle 3, with the same conventions.

\begin{theorem}
For particle 3 at the point $\textrm{P}_{3}$ with position vector $\vct{r}_{3}$, relative to $O$ (the mass centre of particle 1, at the point $\textrm{P}_{1}$ and particle 2, at the point $\textrm{P}_{2}$, which are performing 2-body motion under their mutual gravitational attraction) its equation of motion reads:

$$\ddot{\vct{r}}_{3}=\mu_{1}\frac{\vct{r}_{1}-\vct{r}_{3}}{|\vct{r}_{1}-\vct{r}_{3}|^3}+\mu_{2}\frac{\vct{r}_{2}-\vct{r}_{3}}{|\vct{r}_{2}-\vct{r}_{3}|^3}$$

\npgni and a necessary and sufficient condition for the angular momentum of particle 3, $\vct{h}_{3}$ to be constant is that $|\vct{r}_{1}-\vct{r}_{3}|=|\vct{r}_{2}-\vct{r}_{3}|,$ which requires $r_{3}>2^{-1}(r_{2}-r_{1})$ i.e. $\bigtriangleup\textrm{P}_{1}\textrm{P}_{2}\textrm{P}_{3}$ is isosceles. In this case if the motion is one of rotation and scaling in line with the 2-body problem as above\vspace{-5mm}

$$\ddot{\vct{r}}_{3}=-\frac{{\tilde\mu}_{3}}{r_{3}^{3}}\vct{r}_{3},$$\vspace{-5mm}

\npgni the energy $E_{3}=2^{-1}\left(\dot{r_{3}}^2+\dfrac{h_{3}^2}{r_{3}^2}-
\dfrac{{\tilde\mu}_{3}}{r_{3}}\right)$ and the Hamilton-Lenz-Runge vector,

\npgni $\vct{A}_{3}=\dot{\vct{r}}_{3}\times\vct{h}_{3}+{\tilde\mu}_{3}r_{3}^{-1}\vct{r}_{3},$ are constant if and only if ${\tilde\mu}_{3}$ is constant

\npgni where ${\tilde\mu}_{3}=\dfrac{\mu_{1}+\mu_{2}}{k^3}$, $k=\dfrac{|\vct{r}_{1}-\vct{r}_{3}|}{r_{3}}=\left(1-\dfrac{4\mu_{1}\mu_{2}\cos^2\alpha_{0}}{(\mu_{1}+\mu_{2})^2}\right)^{-\frac{1}{2}}$, $\alpha_{0}$ being the constant

\npgni equal angle of $\bigtriangleup\textrm{P}_{1}\textrm{P}_{2}\textrm{P}_{3}$.

\end{theorem}

\begin{proof}
The angular momentum of particle 3 is $\vct{h}_{3}=\vct{r}_{3}\times\dot{\vct{r}}_{3}$. So

$$\dot{\vct{h}}_{3}=\vct{r}_{3}\times\ddot{\vct{r}}_{3}=\vct{r}_{3}\times\left(\frac{\mu_{1}\vct{r}_{1}}{|\vct{r}_{1}-\vct{r}_{3}|^3}+\frac{\mu_{2}\vct{r}_{2}}{|\vct{r}_{2}-\vct{r}_{3}|^3}\right)$$

\npgni and, if $\hat{\vct{r}}_{2}.\hat{\vct{r}}_{3}=\cos\tilde\alpha$, $\tilde\alpha=\textrm{A}\hat{O}\textrm{J}$,

$$\dot{\vct{h}}_{3}=r_{3}\sin\tilde\alpha\left(-\frac{\mu_{2}r_{2}}{|\vct{r}_{2}-\vct{r}_{3}|^3}+\frac{\mu_{1}r_{1}}{|\vct{r}_{1}-\vct{r}_{3}|^3}\right)$$

\npgni and, since $\mu_{1}\vct{r}_{1}+\mu_{2}\vct{r}_{2}=\vct{0}$, $\mu_{2}r_{2}=\mu_{1}r_{1}$ and the first result follows.

\npgni Now, if $|\vct{r}_{1}-\vct{r}_{3}|=|\vct{r}_{2}-\vct{r}_{3}|$, from $\mu_{1}\vct{r}_{1}+\mu_{2}\vct{r}_{2}=\vct{0}$, the second result follows. Elementary trigonometry completes the proof. 

\end{proof}

\begin{corollary}

\npgni The only way for the $\bigtriangleup\textrm{P}_{1}\textrm{P}_{2}\textrm{P}_{3}$ to be isosceles and to describe the above motion with $k$ and $\alpha_{0}$ constant and $e\ne0$ is for the triangle to be equilateral.

\end{corollary}

\begin{proof}

\npgni A simple calculation yields $\hat{\vct{A}}_{\textrm{i}}.\hat{\vct{r}}_{\textrm{i}}(0)=\dfrac{1}{\tilde{\mu}_{\textrm{i}}e_{\textrm{i}}}\left(-\dfrac{h_{\textrm{i}}^2}{r_{\textrm{i}}(0)}+\tilde{\mu}_{\textrm{i}}\right)$, for $\textrm{i}=1,2,3$.

\npgni Considering the identity for $\textrm{i}=1$ and $\textrm{i}=3$ and the fact that by hypothesis $e_{1}=e_{3}=e$, the common value of the eccentricity, together with $\dot{\theta}_{1}=\dot{\theta}_{3}$ gives

$$\frac{(r_{3}(0))^2}{\tilde{\mu}_{3}}=\frac{(r_{1}(0))^2}{\tilde{\mu}_{1}},$$

\npgni where $\tilde{\mu}_{1}=\dfrac{\mu_{2}^3}{(\mu_{1}+\mu_{2})^2}$, $\tilde{\mu}_{2}=\dfrac{\mu_{1}^3}{(\mu_{1}+\mu_{2})^2}$ from the 2-body problem and $\tilde{\mu}_{3}=\dfrac{(\mu_{1}+\mu_{2})}{k^3}$

\npgni from the above. Together these yield

$$k=\dfrac{(\mu_{1}+\mu_{2})r_{1}(0)}{\mu_{2}r_{3}(0)}=\dfrac{r(0)}{r_{3}(0)},$$

\npgni by definition of $O$. However, we have assumed that $k=\dfrac{a_{0}}{r_{3}}$.

\end{proof}

\npgni Nevertheless there is a solution in which $k$ and $\alpha_{0}$ vary in time as we see next.

\begin{theorem}
(Isosceles Triangle Equation)

\npgni The equation of motion for $\textrm{P}_{3}$, when $h_{3}=L$ is constant, and $(')=\dfrac{d}{d\theta}$ reduces to:

$$\rho''-\dfrac{2}{r}r'\rho'-\dfrac{L^2r^4}{h^2\rho^3}+\dfrac{(\mu_{1}+\mu_{2})r^4}{h^2}\left(\rho^2+\dfrac{\mu_{1}\mu_{2}}{(\mu_{1}+\mu_{2})^2}r^2\right)^{-\frac{3}{2}}\rho=0,$$

\npgni where $\dfrac{l}{r}=1+e\cos\theta$, $r=|\vct{r}_{1}-\vct{r}_{2}|$, $h=r^2\dot{\theta}$, $l=\dfrac{h^2}{(\mu_{1}+\mu_{2})}$.

\npgni This equation has a unique solution in a neighbourhood of the Lagrange equilateral triangle solution for which

$$\rho^2=\dfrac{1}{4}\left(\dfrac{\mu_{1}-\mu_{2}}{\mu_{1}+\mu_{2}}\right)^2r^2+\dfrac{3}{4}r^2$$

\npgni i.e. $\rho=\tilde{\mu}r$ with $\tilde{\mu}^2=\dfrac{\mu_{1}^2+\mu_{2}^2+\mu_{1}\mu_{2}}{(\mu_{1}+\mu_{2})^2}$. Writing $\rho=\tilde{\mu}a_{0}+\epsilon$, for small $\epsilon$, $a_{0}$ the Lagrange

\npgni solution, for $e=0$, yields $\epsilon''+\tilde{k}^2\epsilon=0$, $\tilde{k}^2=4-3\tilde{\mu}^2>0$, and for $e\sim0$, $\dfrac{l}{r}=1+e\cos{\theta}$, $\rho=\tilde{\mu}r+\epsilon$, $\epsilon=\textrm{exp}(-e\cos{\theta})z$, where $z$ is the 
Mathieu function:

$$\dfrac{d^2z}{d\chi^2}+4(4-3\tilde{\mu}^2+3\tilde{\mu}e\cos{2\chi})z=0,\;\;\chi=\frac{\theta}{2},$$

\npgni for all choices of $\mu_{1}$ and $\mu_{2}$.

\npgni [Note: defining $\dot{\phi}$ by $\rho^2\dot{\phi}=L$, a suitable constant value of angular momentum, we have an isosceles triangle solution of the restricted 3-body problem in which the triangle rotates and pulsates for the circular orbital case].

\end{theorem}

\begin{proof}

The requirement that the $\bigtriangleup\textrm{P}_{1}\textrm{P}_{2}\textrm{P}_{3}$ be isosceles requires $\rho^2\dot{\phi}=L$ and 

$$\dfrac{d^2\rho}{d\phi^2}-\dfrac{L^2}{\rho^3}=-(\mu_{1}+\mu_{2})\left(\rho^2+\dfrac{\mu_{1}\mu_{2}}{(\mu_{1}+\mu_{2})^2}r^2\right)^{-\frac{3}{2}}\hspace{-2mm}\rho,$$

\npgni $L$ constant and the equal side length $a$ being

$$a^2=\rho^2+\dfrac{\mu_{1}\mu_{2}}{(\mu_{1}+\mu_{2})^2}r^2.$$

\npgni To obtain the isosceles triangle equation we merely change the independent variable from $\phi$ to $\theta$, where $r^2\dot{\theta}=h$, which comes from the 2-body equation. The rest of the proof is a simple computation.

\end{proof}

\begin {corollary} (Isosceles Triangle Orbital Equation for $e\sim0$)

\npgni When the eccentricity of the 2-body orbit, $e=0$, for the isosceles triangle solution the equation of the asteroid's orbit is given by the quadrature:

$$\int\frac{du}{\sqrt{E-\dfrac{L^2u^2}{2}+(\mu_{1}+\mu_{2})u\left(1+\dfrac{\mu_{1}\mu_{2}a_{0}^2}{(\mu_{1}+\mu_{2})^2}u^2\right)^{-\frac{1}{2}}}}=\pm\sqrt{\dfrac{2}{L^2}}\int d\phi,$$

\npgni where $u=\dfrac{1}{\rho}$, $E$ being the energy and $a_{0}=|\vct{r}_{2}-\vct{r}_{1}|$. For $E<0$ the isosceles triangle orbit is bounded away from $\rho=0$ and $\rho=\infty$. The orbit is closed if and only if

\npgni $$\Phi=\sqrt{\dfrac{L^2}{2}} \int_{u_{\textrm{min}}}^{u_{\textrm{max}}}\dfrac{du}{\sqrt{E-V_{\textrm{eff}}(u)}}=\dfrac{2\pi m}{n},\;\;m,n\in\mathbb{Z}.$$

\npgni In this case there is a unique circular isosceles triangle solution which is stable for small $\mu_{2}$.

\end{corollary}

\begin{proof}

The effective potential here is

$$V_{\textrm{eff}}=\dfrac{L^2}{2\rho^2}-\dfrac{\mu_{1}+\mu_{2}}{(\rho^2+c^2)^\frac{1}{2}},\;\;\;\textrm{where}\;\;\;c^2=\dfrac{\mu_{1}\mu_{2}}{(\mu_{1}+\mu_{2})^2}a_{0}^2,$$

\npgni $a_{0}$ the radius of the 2-body orbit. So as $\rho\searrow0$, $V_{\textrm{eff}}(\rho)\nearrow\infty$ and as $\rho\nearrow\infty$, $V_{\textrm{eff}}(\rho)\nearrow0$. Further

$$V'_{\textrm{eff}}(\rho)=-\dfrac{L^2}{\rho^3}+\dfrac{(\mu_{1}+\mu_{2})\rho}{(\rho^2+c^2)^\frac{3}{2}}=0$$

\npgni if and only if $\dfrac{\rho^4}{(\rho^2+c^2)^\frac{3}{2}}=\dfrac{L^2}{\mu_{1}+\mu_{2}}.$ This last equation has a solution $\rho_{0}>0$ if and only

\npgni if, for $X=\rho^2$, $Y_{1}(X)=Y_{2}(X)$, where

$$Y_{1}(X)=\dfrac{(\mu_{1}+\mu_{2})^{\frac{1}{2}}}{L}X\;\;\;,\;\;\;Y_{2}(X)=(X^2+c^2)^{\frac{3}{4}},$$

\npgni $Y_{2}$ being convex downwards for $\rho>0$, $Y'_{2}>0$, $Y''_{2}<0$. It is then easy to see graphically that $V_{\textrm{eff}}(\rho)$ has a unique minimiser $\rho=\rho_{0}$ if $E<0$. It follows that, if $V_{\textrm{eff}}(\rho_{0})<E<0$, the equation, $V_{\textrm{eff}}(\rho)=E$, has two solutions $\rho_{\textrm{min}}$, $\rho_{\textrm{max}}$, $0<\rho_{\textrm{min}}<\rho_{\textrm{max}}$, proving the first part of the result. Assuming $L^2$, $c^2$, $(\mu_{1}+\mu_{2})$ are such that $V''_{\textrm{eff}}>0$ implies there is a stable circular orbit at $\rho=\rho_{0}$. The last result is standard (See e.g. Arnold Ref.[3]).
\end{proof}

\npgni For non constant $u$ one can evaluate the integral on the l.h.s. for small $\dfrac{\mu_{1}\mu_{2}a_{0}^2}{(\mu_{1}+\mu_{2})^2}$ in terms of Weierstrass' elliptic function, $\wp$, giving an application to the Hildan asteroids as we show next.

\npgni Letting $\mu_{2}\rightarrow0$ in Corollary 3.2.1 we obtain for the isosceles triangle solution

$$\sqrt{\frac{2}{L^2}}\int d\phi=\int\frac{du}{\sqrt{E-V_{0}}},$$

\npgni where $\rho=\dfrac{1}{u}$ and $V_{0}=\dfrac{L^2}{2}u^2-\mu_{1}u$, $E$ being a negative constant.

\npgni It is easy to deduce from this result that

$$r_{3}=\rho=\dfrac{\tilde l}{1+\tilde e\cos(\phi-\phi_{0})},\;\;\tilde l=\dfrac{L^2}{\mu_{1}},\;\;\tilde e=\sqrt{1+\dfrac{2L^2E}{\mu_{1}^2}},$$

\npgni the 1-body equation of our asteroid, particle 3 in the restricted 3-body problem, when $\mu_{2}=0$.

\npgni The Hildans are in $\dfrac{3}{2}$ resonance with Jupiter and at aphelion on Jupiter's circular orbit of radius $r_{0}$. So for the paradigm Hildan, from Kepler's $3^{\textrm{rd}}$ Law, $\dfrac{2}{3}=(1+\tilde e)^{-\frac{3}{2}}$,

$$r_{0}=\tilde a(1+\tilde e),\;\;\;\;\;\tilde l=\tilde a(1-\tilde e^2),\;\;\;\;\;\tilde e=\sqrt{1+\dfrac{2L^2E}{\mu_{1}^2}},$$\vspace{-10mm}

\npgni where\vspace{-5mm}

$$\tilde e=\left(\dfrac{2}{3}\right)^{-\frac{2}{3}}\hspace{-1mm}-1\sim0.307.$$

\npgni In fact according to the latest data the eccentricities of the Hildans are all less than $0.3$. A simple computation confirms that for this restricted 3-body (Sun-Jupiter-Hildan) problem, $r_{3}>\dfrac{r_{2}-r_{1}}{2}$, $\rho=r_{3}$. So the isosceles triangle solution is available for the Hildans. We now compute the equation of this isosceles triangle orbit as $\mu_{2}\sim0$.

\npgni From the above, correct to $1^{\textrm{st}}$ order,

$$\dfrac{\sqrt{2}}{L}(\phi-\phi_{0})=\int\frac{du}{\sqrt{Q(u)}}\stackrel{\text{def}}{=}z,$$\vspace{-10mm}

\npgni where

$$Q(u)=-\dfrac{\mu_{2}r_{0}^2}{2}u^3-\dfrac{L^2}{2}u^2+(\mu_{1}+\mu_{2})u+E,$$

\npgni $Q(u)$ being the quartic, $Q(u)=a_{0}u^4+4a_{1}u^3+6a_{2}u^2+4a_{3}u+a_{4}$, with,

$$a_{0}=0,\;\;\;\;\;a_{1}=-\dfrac{\mu_{2}r_{0}^2}{8},\;\;\;\;\;a_{2}=-\dfrac{L^2}{12},\;\;\;\;\;a_{3}=\dfrac{\mu_{1}+\mu_{2}}{4},\;\;\;\;\;a_{4}=E.$$

\npgni This enables us to use the celebrated Weierstrass formula:

$$u-u_{0}=\dfrac{1}{4}Q'(u_{0})\Big[\wp(z;g_{2},g_{3})-\dfrac{1}{24}Q''(u_{0})\Big]^{-1},$$

\npgni for quartic invariants $g_{2}$ and $g_{3}$, $\wp$ the Weierstrass elliptic function and $u=u_{0}$ being any root of $Q(u)=0$. (See Whittaker and Watson Ref.[36] pp. 452-3, where their argument still works for cubics i.e. when $a_{0}=0$, as long as we work with the 
appropriate root of the resulting cubic). Correct to first order in $\mu_{2}$,

$$g_{2}=\dfrac{L^4}{48}+\dfrac{\mu_{1}\mu_{2}r_{0}^2}{8},\;\;\;\;\;g_{3}=\dfrac{\mu_{1}\mu_{2}L^2r_{0}^2}{192}+\dfrac{L^6}{1728}.$$

$$Q'(u)=-\dfrac{3}{2}\mu_{2}r_{0}^2u^2-L^2u+(\mu_{1}+\mu_{2}),\;\;\;\;\;Q''(u)=-(3\mu_{2}r_{0}^2u+L^2).$$

\npgni We have proved:-

\begin{theorem}

\npgni Choosing $u_{0}=\dfrac{1}{\rho_{\textrm{max}}}$, $\rho_{\textrm{max}}$ the maximum value of $\rho$ on the asteroid isosceles triangle orbit, we obtain the equation of this orbit correct to first order in $\mu_{2}$:

$$\dfrac{1}{\rho}-\dfrac{1}{\rho_{\textrm{max}}}=\dfrac{1}{4}\left(\mu_{1}+\mu_{2}-\dfrac{3\mu_{2}r_{0}^2}{2\rho_{\textrm{max}}^2}-\dfrac{L^2}{\rho_{\textrm{max}}}\right)\Bigg[\wp\left(\dfrac{\sqrt{2}}{L}(\phi-\phi_{0});g_{2},g_{3}\right)+\dfrac{1}{24}\left(\dfrac{3\mu_{2}r_{0}^2}{\rho_{\textrm{max}}}+L^2\right)\Bigg]^{-1}\hspace{-2mm},$$

\npgni where the roots of the cubic, $C(\rho)$, for sufficiently small $\mu_{2}$, are real $\rho_{\textrm{max}}$, $\rho_{\textrm{min}}$, $\rho_{3}$,

$$C(\rho)=\rho^3+\dfrac{\mu_{1}+\mu_{2}}{E}\rho^2-\dfrac{L^2}{2E}\rho-\dfrac{\mu_{2}r_{0}^2}{2E}=0,$$

\npgni with $\rho_{\textrm{max}}>\rho_{\textrm{min}}>0>\rho_{3}$, and on the orbit $\rho_{\textrm{min}}<\rho<\rho_{\textrm{max}}$.

\end{theorem}

\npgni Needless to say in the limit as $\mu_{2}\rightarrow0$, $\rho_{\textrm{min}}$ and $\rho_{\textrm{max}}$ converge to the roots of the quadratic characterising the Keplerian elliptic orbit for the 1-body problem:

$$q(\rho)=\rho^2+\dfrac{\mu_{1}}{E}\rho-\dfrac{L^2}{2E},$$

\npgni $E<0$, namely $\tilde a(1\pm\tilde e)$,  $\tilde a=\dfrac{r_{0}}{1+\tilde e}$,  $\tilde e=\sqrt{1+\dfrac{2L^2E}{\mu_{1}^2}}=\left(\dfrac{2}{3}\right)^{-\frac{2}{3}}\hspace{-1mm}-1$, $\rho_{3}\rightarrow\infty$ for the paradigm Hildan asteroid.

\npgni Here $E$ and $L$ are given by\hspace{2mm} $-\dfrac{\mu_{1}}{2E}=\left(\dfrac{2}{3}\right)^{\frac{2}{3}}r_{0}$,\hspace{3mm}$\dfrac{L^2}{\mu_{1}}=\left(2-\left(\dfrac{2}{3}\right)^{-\frac{2}{3}}\right)r_{0}$.

\npgni For isosceles triangle solutions of the restricted 3-body problem, as long as 
$r_{3}>\dfrac{r_{2}-r_{1}}{2}$, the above theorem generalises.

\begin{corollary}

\npgni The asymptotic expansion as $\mu_{2}\sim0$ of $\wp$ in the last theorem, for the above $g_{2}$ and $g_{3}$, positive, reduces to:-

$$\wp(z;g_{2},g_{3})=-(c_{0}+\delta c)+\dfrac{3(c_{0}+\delta c)}{\sin^2(\sqrt{3(c_{0}+\delta c)}(z_{0}+\delta z))},$$

\npgni to first order, where $c_{0}=\dfrac{L^2}{24}$, $c_{0}+\delta c=c_{0}\left(1+\dfrac{3\mu_{1}\mu_{2}r_{0}^2}{L^4}\right)$, $\sqrt{3c}z_{0}=\dfrac{1}{2}(\phi-\phi_{0})$,

\npgni $(\phi-\phi_{0})$ being the polar angle of particle 3 measured from $\rho_{\textrm{max}}$, where $\phi=\phi_{0}$,

$$\delta z=z-z_{0}=\dfrac{\delta g_{2}}{144c_{0}^2}\int\dfrac{\sin^4(\sqrt{3c}z)dz}{\cos^2(\sqrt{3c}z)}+\dfrac{(\delta g_{3}-c_{0}\delta g_{2})}{144c_{0}^2}\int\dfrac{\sin^6(\sqrt{3c}z)dz}{\cos^2(\sqrt{3c}z)},$$

\npgni $\delta g_{2}=\dfrac{\mu_{1}\mu_{2}r_{0}^2}{8}$, $\delta g_{3}=\dfrac{\mu_{1}\mu_{2}L^2r_{0}^2}{192}$ ; and\vspace{3mm}

$$\int\dfrac{\sin^4x}{\cos ^2x}dx=\dfrac{\sin^3x}{\cos x}-\dfrac{3}{2}\left(x-\dfrac{1}{2}\sin 2x\right),$$\vspace{1mm}

$$\int\dfrac{\sin^6x}{\cos ^2x}dx=\dfrac{\sin^5x}{\cos x}-\dfrac{5}{4}\left(\dfrac{3}{2}x-\sin 2x+\dfrac{1}{8}\sin 4x\right).$$

\end{corollary}

\begin{proof}

\npgni The point is that the discriminant, $\Delta$, of $(4\wp^3-g_{2}\wp-g_{3})$ is $\Delta=g_{2}^3-27g_{3}^2$, i.e.

$$\Delta=\left(\dfrac{L^4}{48}+\dfrac{\mu_{1}\mu_{2}r_{0}^2}{8}\right)^3-27\left(\dfrac{\mu_{1}\mu_{2}L^2r_{0}^2}{192}+\dfrac{L^6}{1728}\right)^2>0$$

\npgni reduces to $\Delta=\dfrac{\mu_{1}^2\mu_{2}^2r_{0}^4L^4}{4096}+\dfrac{\mu_{1}^3\mu_{2}^3r_{0}^3}{512}=0$, to first order in $\mu_{2}$ and $g_{2}>0$, $g_{3}>0$.

\npgni So there is a solution $c=c_{0}+\delta c$ of the pair of equations, correct to first order,

$$12c^2=\dfrac{L^4}{48}+\dfrac{\mu_{1}\mu_{2}r_{0}^2}
{8}=g_{2},\;\;\;8c^3=\dfrac{\mu_{1}\mu_{2}L^2r_{0}^2}{192}+\dfrac{L^6}{1728}=g_{3},$$

\npgni namely $c=c_{0}+\delta c$ above. So for our purposes the cubic in $\wp$ has still got a pair of equal roots $-c$ and $-c$, the remaining root being $2c$ and we can find

$$\int \dfrac{d\wp}{\sqrt{4\wp^3-g_{2}\wp-g_{3}}}=\int dz=z_{0}+\delta z$$

\npgni by expanding in $\delta g_{2}$ and $\delta g_{3}$ if you use the new integration variable $w$, where

$$\wp=-c+\dfrac{3c}{\sin^2(\sqrt{3c}w)},\;\;\;\textrm{a simple enough expression}.$$

\end{proof}\vspace{-8mm}

\npgni It follows that for $\tilde{e}$ the eccentricity and $\tilde{l}$ the semi-latus rectum of the ellipse corresponding to $\mu_{2}=0$, correct to first order in $\mu_{2}$ as $\mu_{2}\sim0$, $u=\dfrac{1}{\rho}$ and $u_{0}=u_{\textrm{min}}$,\vspace{3mm}

$$u-u_{0}=\dfrac{2(\tilde{e}-\mu_{1}^{-1}\mu_{2})\sin^2(\sqrt{3(c+\delta c)}(z_{0}+\delta z))}{\tilde{l}+\mu_{1}^{-1}\mu_{2}r_{0}\left(\dfrac{r_{0}}{\tilde{l}}+\sin^2\left(\dfrac{\phi-\phi_{0}}{2}\right)\right)},$$\vspace{-5mm}

\npgni where $r_{0}=r_{\textrm{max}}=\dfrac{\tilde{l}}{1-\tilde{e}}$, $u_{0}=\dfrac{1}{r_{0}}$, always assuming $E<0$. This remarkable aysmptotic formula details the gravitational effects of the second body e.g. Jupiter in its orbit. Unfortunately the simplifying assumption of $\Delta=0$ means that it blows up e.g. when $(\phi-\phi_{0})=\pi$. The full formula involving $\wp$ does not have this problem when $\Delta>0$.\vspace{-3mm}

\npgni The following result, for the motion which is bounded away from $0$ and $\infty$, for which we give a simple proof, has none of these problems and gives the equations of the isosceles triangle orbit in all generality, as long as the third particle mass is negligible.

\begin{theorem}
\npgni Assuming that the orbit lies in the region $r_{3}>\dfrac{1}{2}(r_{2}-r_{1})$, given that $0<w<\sqrt{\dfrac{\mu_{1}}{\mu_{2}r_{2}^2}}$ i.e. $0<\dfrac{1}{\sqrt{\rho^2+\dfrac{\mu_{2}}{\mu_{1}}r_{2}^2}}<\sqrt{\dfrac{\mu_{1}}{\mu_{2}r_{2}^2}}$, the equation of the orbit of particle 3 is given by

$$\hspace{-5mm}\dfrac{1}{2}\int\dfrac{dw}{\left(1-\sqrt{\dfrac{\mu_{2}r_{2}^2}{\mu_{1}}}w\right)\sqrt{E+(\mu_{1}+\mu_{2})w-\left(\dfrac{L^2}{2}+\dfrac{\mu_{2}}{\mu_{1}}r_{2}^2E\right)w^2-\dfrac{(\mu_{1}+\mu_{2})}{\mu_{1}}\mu_{2}r_{2}^2w^3}}$$

$$\;\;\;\;+\dfrac{1}{2}\int\dfrac{dw}{\left(1+\sqrt{\dfrac{\mu_{2}r_{2}^2}{\mu_{1}}}w\right)\sqrt{E+(\mu_{1}+\mu_{2})w-\left(\dfrac{L^2}{2}+\dfrac{\mu_{2}}{\mu_{1}}r_{2}^2E\right)w^2-\dfrac{(\mu_{1}+\mu_{2})}{\mu_{1}}\mu_{2}r_{2}^2w^3}}$$

$$\hspace{50mm}=\sqrt{\dfrac{2}{L^2}}\int d\phi,$$

\npgni $w^{-1}=|\vct{r}_{3}-\vct{r}_{1}|$ being the equal side length in the isosceles $\bigtriangleup P_{1}P_{2}P_{3}$.

\end{theorem}

\begin{proof}

\npgni The equation of the orbit satisfies

$$\dfrac{L^2}{\rho^4}\left(\dfrac{d\rho}{d\phi}\right)^2+\dfrac{L^2}{\rho^2}-2(\mu_{1}+\mu_{2})\left(\rho^2+\dfrac{\mu_{2}}{\mu_{1}}r_{2}^2\right)^{-\frac{1}{2}}=2E.$$\vspace{-5mm}

\npgni Using the substitution $w=\dfrac{1}{\sqrt{\rho^2+\dfrac{\mu_{2}}{\mu_{1}}r_{2}^2}}$ gives the solution.

\end{proof}\vspace{-5mm}

\npgni When the discriminant of the cubic is such that its roots are real, the integrals on the l.h.s. can be expressed in terms of elliptic integrals of the first and third kind provided the cubic within the integrand factorises appropriately. In the nomenclature of Gradshteyn and Ryzhik, (see Ref.[15]),

$$\hspace{-90mm}\int_{b}^{w}\dfrac{dx}{(x-d)\sqrt{(a-x)(x-b)(x-c)}}$$

$$=\dfrac{2}{(c-d)(b-d)\sqrt{a-c}}\Bigg[(c-b)\Pi\left(K,\left(\dfrac{c-d}{b-d}\right)p^2,p\right)+(b-d)F(K,p)\Bigg],$$

\npgni where $p=\sqrt{\dfrac{a-b}{a-c}}$, $K=\arcsin\sqrt{\dfrac{(a-c)(w-b)}{(a-b)(w-c)}}$ and\vspace{3mm}

$$F(\phi,k)=\int_{0}^{\phi}\dfrac{d\alpha}{\sqrt{1-k^2\sin^2\alpha}}=\int_{0}^{\sin\phi}\dfrac{dx}{\sqrt{(1-x^2)(1-k^2x^2)}},$$

$$\Pi(\phi,n,k)=\int_{0}^{\phi}\dfrac{d\alpha}{(1-n\sin^2\alpha)\sqrt{1-k^2\sin^2\alpha}}=\int_{0}^{\sin\phi}\dfrac{dx}{(1-nx^2)\sqrt{(1-x^2)(1-k^2x^2)}}.$$

\npgni Following methods introduced by Legendre this result is easily obtained using the fractional substitution $y^2=\dfrac{(a-c)(x-b)}{(a-b)(x-c)}$, the details being left as an exercise.

\npgni (See also McKean and Moll Ref.[18] or Abramowitz and Stegun Ref.[2]).

\npgni Here $a,\;b,\;c$ are the roots of our cubic above, $E<0$ by assumption with $d\ne b$, $a>b>0>c$. This follows by considering the product of the roots and the sum of the product of roots in pairs, both being negative. $a,\;b,\;c$ can be calculated using Vieta's formula assuming the discriminant condition for real roots:

$$(2B^3-9ABC+27A^2D)^2<4(B^2-3AC)^3,$$

\npgni where $A=-\dfrac{(\mu_{1}+\mu_{2})\mu_{2}r_{2}^2}{\mu_{1}}$, $B=-\left(\dfrac{L^2}{2}+\dfrac{\mu_{2}r_{2}^2E}{\mu_{1}}\right)$, $C=\mu_{1}+\mu_{2}$ and $D=E$.

\npgni This can be expressed in powers of $L^2$ and $E$:

$$C(X)=4\gamma YX^3+(12\gamma Y^2+1)X^2+4(3\gamma Y^2+5)YX+4\gamma(Y^2-\gamma^{-1})^2>0,$$

\npgni where $X=\dfrac{L^2}{2}$, $Y=\dfrac{\mu_{2}r_{2}^2}{\mu_{1}}E$ and $\gamma=\dfrac{\mu_{1}}{(\mu_{1}+\mu_{2})^2\mu_{2}r_{2}^2}$. Viewing the above expression\vspace{-3mm}

\npgni $C(X)$ as a cubic in $X$, with $C(0)>0$ and $Y<0$, we easily deduce that provided $Y^2\ne\gamma^{-1}$ there exists an open set for which $X>0$ and the cubic is positive definite, thus ensuring the existence of the roots $a,\;b,\;c$.

\npgni The limits of the motion are $a\ge w\ge b$, $w=\dfrac{1}{\sqrt{\rho^2+\dfrac{\mu_{2}}{\mu_{1}}r_{2}^2}}$. Moreover the orbit is closed if and only if

$$\hspace{-5mm}\dfrac{L}{2\sqrt{2}}\int_{b}^a\dfrac{dw}{\left(1-\sqrt{\dfrac{\mu_{2}r_{2}^2}{\mu_{1}}}w\right)\sqrt{E+(\mu_{1}+\mu_{2})w-\left(\dfrac{L^2}{2}+\dfrac{\mu_{2}}{\mu_{1}}r_{2}^2E\right)w^2-\dfrac{(\mu_{1}+\mu_{2})}{\mu_{1}}\mu_{2}r_{2}^2w^3}}$$

$$\;+\dfrac{L}{2\sqrt{2}}\int_{b}^a\dfrac{dw}{\left(1+\sqrt{\dfrac{\mu_{2}r_{2}^2}{\mu_{1}}}w\right)\sqrt{E+(\mu_{1}+\mu_{2})w-\left(\dfrac{L^2}{2}+\dfrac{\mu_{2}}{\mu_{1}}r_{2}^2E\right)w^2-\dfrac{(\mu_{1}+\mu_{2})}{\mu_{1}}\mu_{2}r_{2}^2w^3}}$$

$$\hspace{100mm}=\Phi=\dfrac{2\pi m}{n},\;\;m,n\in\mathbb{Z}.$$\vspace{-10mm}

\npgni For the whole orbit to be described by this isosceles triangle solution we need

\npgni $\dfrac{1}{a}>\dfrac{1}{2}(r_{2}+r_{1})$ and then we have the hidden constant: $\dfrac{|\textrm{AJ}|}{|\textrm{AS}|}=\dfrac{|\vct{r}_{3}-\vct{r}_{2}|}{|\vct{r}_{3}-\vct{r}_{1}|}=1$, on the entire orbit.\vspace{-3mm}

\npgni \textbf{Remarks}\vspace{-3mm}

\npgni 1. For the Trojan asteroid system, $\tilde{k}=1.00143$.

\npgni 2. When Jupiter's orbit is circular there is no energy transfer whatever (c.f. Theorem 3.4).

\npgni 3. See the Appendix for Kepler's $4^{\textrm{th}}$ Law for 3-body problems when $a\sim b$.

\npgni 4. Of course Trojan asteroids, suitably perturbed, provide an example of this result for small $\mu_{2}$. There could be others hopefully. Needless to say constants of the motion here include the ratios of the sides of the triangle be it equilateral or isosceles. This result is the most general because the triangle is isosceles if and only if $\vct{h}_{3}$ is conserved and $(\ast)$ and $(\ast\ast)$ depend on this for their validity.

\subsection{Linearised Problem for Equlateral Triangle Elliptical Orbits}\vspace{-5mm}

\npgni From the above analysis, setting $\omega=\dot{\theta_{0}}(t)$, the common value of $\dot{\theta_{\textrm{i}}}(t)$, $\textrm{i}=1,2,3$, for $e\ne0$, working in rotating axes $O\tilde{x}$, $O\tilde{y}$; $O\tilde{x}$ being parallel to $\overset{\rightharpoonup}{O\textrm{J}}$, the Hamiltonian is Jacobi's integral\vspace{-5mm}

$$K=2^{-1}(\tilde{p}_{x}^2+\tilde{p}_{y}^2)+\dot\theta_{0}(t)(\tilde y\tilde p_{x}-\tilde x\tilde p_{y})+\tilde V,\;\;\;\textrm{where}\;\;\;\tilde V=-\dfrac{\mu_{1}}{|\textrm{SA}|}-\dfrac{\mu_{2}}{|\textrm{JA}|}.$$\vspace{-5mm}

\npgni Observing that $\dfrac{\partial K}{\partial t}=\ddot{\theta}(t)(\tilde y\tilde p_{x}-\tilde x\tilde p_{y})$ we see that when $e\ne0$, $K$ is not a constant.

\npgni However, setting $\tilde{\vct{A}}=\dot\theta(t)(-\tilde y,\tilde x)$,

$$K=2^{-1}(\tilde {\vct{p}}-\tilde {\vct{A}})^{2}+\tilde V-\frac{\tilde {\vct{A}}^2}{2}=2^{-1}(\tilde {\vct{p}}-\tilde {\vct{A}})^{2}+\tilde W,$$\vspace{-8mm}

\npgni $\tilde W=\tilde V-\dfrac{\tilde {\vct{A}}^2}{2}$, in the rotating frame.

\npgni Therefore, we obtain, if $\overset{\rightharpoonup}{O\textrm{A}}=\vct r$ in the rest frame, in the rotating frame,

$$\ddot{\tilde{\vct r}}=-\textrm{grad}\tilde{W}-\frac{\partial\tilde{\vct{A}}}{\partial t}+\dot{\tilde{\vct{r}}}\times\textrm{curl}\tilde{\vct{A}},$$\vspace{-8mm}

\npgni where $\tilde{W}=-\dfrac{\dot\theta_{0}^2}{2}(\tilde{x}^2+\tilde{y}^2)+\tilde{V}$. However, when $\tilde{\vct{r}}=\overset{\rightharpoonup}{O\textrm{L}}=\vct{r}_{0}$ in the rest frame our equation is satisfied if $\dot{\theta_{0}}=\omega$, so it is natural to linearise about this equilibrium point by writing, $\tilde{\vct{r}}=\vct{r}_{0}+\delta\vct{r}_{0}$, where we write $\delta\vct{r}=(\delta(t),\epsilon(t))$ in the rotating coordinates giving,

\npgni \;\;\;\;\;\;\;\;\;\;\;$\ddot{\vct{r}}_{0}+\delta\ddot{\tilde{\vct{r}}}=-\textrm{grad}\tilde{W}{\rceil}_{0}-\delta\vct{r}.\textrm{grad}\tilde{W}{\rceil}_{0}-\dfrac{\partial}{\partial t}(\tilde{\vct{A}}_{0}+(\delta\vct{r}_{0}.\boldsymbol{\nabla})\tilde{\vct{A}}{\rceil}_{0})$

\npgni \;\;\;\;\;\;\;\;\;\;\;\;\;\;\;\;\;\;\;\;\;\;\;\;\;\;\;\;\;\;\;\;\;\;\;\;\;\;\;\;\;\;\;\;\;\;\;\;\;\;\;\;\;\;\;\;\;\;\;\;\;\;\;\;$+(\dot{\vct{r}_{0}}+\delta\dot{\tilde{\vct{r}}})\times(\textrm{curl}\tilde{\vct{A}}
{\rceil}_{0}+(\delta\vct{r}_{0}.\boldsymbol{\nabla})\textrm{curl}\tilde{\vct{A}}{\rceil}_{0})$,

\npgni where ${\rceil}_{0}$ means evaluate at $\textrm{L}=\mathscr{L}_{4,5}$. This gives the linearised equation

$$(\ddot{\delta},\ddot{\epsilon})=-\left(\delta\frac{\partial}{\partial x}+\epsilon\frac{\partial}{\partial y}\right)(\textrm{grad}\tilde{W}){\rceil}_{0}+\ddot{\theta}_{0}(t)(\epsilon,-\delta)+2\dot{\theta}_{0}(t)(\dot{\epsilon},-\dot{\delta}).$$

\npgni To solve this equation we use analyticity in $e$ and use the asymptotic series:-

$$\dbinom{\delta}{\epsilon}=\dbinom{\delta_{0}}{\epsilon_{0}}+e\dbinom{\delta_{1}}{\epsilon_{1}}+e^2\dbinom{\delta_{2}}{\epsilon_{2}}+\cdots,$$

\npgni which may even converge in sup norm for almost periodic functions.

\npgni We work only to first order in the eccentricity $e\sim0$ giving (using Jupiter's or $\textrm{P}_{2}\textrm{'}s$ parameters),

$$\dot{\theta_{0}}=\frac{h_{\textrm{J}}}{l_{\textrm{J}}^2}(1+e\cos{\theta_{0}(t)})^2\sim\frac{h_{\textrm{J}}}{l_{\textrm{J}}^2}(1+2e\cos{\theta_{0}(t)}),$$

$$\ddot{\theta_{0}}=\frac{2h_{\textrm{J}}^2}{l_{\textrm{J}}^4}(1+e\cos{\theta_{0}(t)})^3(-e\sin{\theta_{0}(t)})\sim-2\omega^2e\sin{\theta_{0}(t)},$$\vspace{-5mm}

\npgni i.e. correct to first order in $e$, $\ddot{\theta_{0}}\sim-2\omega^2e\sin{\omega t}$. So we obtain,

$$\ddot{\delta}-2\dot{\theta_{0}}\dot{\epsilon}-\frac{3}{4}{\dot\theta_{0}}^2\delta-\Omega^2\epsilon=\epsilon\ddot\theta_{0},$$

$$\ddot{\epsilon}+2\dot{\theta_{0}}\dot{\delta}-\frac{9}{4}{\dot\theta_{0}}^2\epsilon-\Omega^2\delta=-\delta\ddot\theta_{0}.$$\vspace{-8mm}

\npgni Therefore, correct to first order in $e$,

\npgni \;\;\;\;\;\;\;\;\;\;\;$\begin{pmatrix}D^2-\frac{3}{4}\omega^2 & -2\omega D-\Omega^2 \\[0.4em] 2\omega D-\Omega^2 & D^2-\frac{9}{4}\omega^2 \end{pmatrix}\begin{pmatrix}{\delta}\\[0.4em] \epsilon\end{pmatrix}$\vspace{5mm}

$$\;\;\;\;\;\;\;=e\begin{pmatrix}3\omega^2\cos{\theta} & 4\omega\cos{\theta}D-2\omega^2\sin{\theta} \\[0.4em] 2\omega^2\sin{\theta}-4\omega\cos{\theta}D & 9\omega^2\cos{\theta} \end{pmatrix}\begin{pmatrix}{\delta}\\[0.4em] \epsilon\end{pmatrix},$$\vspace{-5mm}

\npgni $D=\dfrac{d}{dt}$, where on bounded time intervals, correct to first order, $e\cos{\theta}\sim e\cos{\omega t}+\textrm{O}(e^2)$\vspace{-3mm}

\npgni and $e\sin{\theta}\sim e\sin{\omega t}+\textrm{O}(e^2)$, giving

\npgni \;\;\;\;\;\;\;\;\;\;\;$\begin{pmatrix}D^2-\frac{3}{4}\omega^2 & -2\omega D-\Omega^2 \\[0.4em] 2\omega D-\Omega^2 & D^2-\frac{9}{4}\omega^2 \end{pmatrix}\begin{pmatrix}{\delta_{1}}\\[0.4em] \epsilon_{1}\end{pmatrix}$\vspace{5mm}

$$\;\;\;\;\;\;\;=\begin{pmatrix}3\omega^2\cos{\omega t} & 4\omega\cos{\omega t}D-2\omega^2\sin{\omega t} \\[0.4em] 2\omega^2\sin{\omega t}-4\omega\cos{\omega t}D & 9\omega^2\cos{\omega t} \end{pmatrix}\begin{pmatrix}{\delta_{0}}\\[0.4em] \epsilon_{0}\end{pmatrix}.$$

\npgni As expected we see that $\dbinom{\delta_{1}}{\epsilon_{1}}$ is almost periodic as is $\dbinom{\delta_{0}}{\epsilon_{0}}$.\vspace{5mm}

\begin{lemma}

By inspection of the last identity, when $e\sim0$, for $\sigma_{\textrm{B}}$ the Bohr spectrum,

$$\sigma_{\textrm{B}}(X)=\sigma_{\textrm{B}}(Y)=\sigma_{\textrm{B}}(\delta)=\sigma_{\textrm{B}}(\epsilon)=\{\pm\alpha,\pm\beta\pm(\omega-\alpha),\pm(\omega+\alpha)\},$$

\npgni where $\sigma_{\textrm{B}}(\delta_{0})=\sigma_{\textrm{B}}(\epsilon_{0})=\{\pm\alpha,\pm\beta\}$. So to find $\dbinom{\delta}{\epsilon}$ all we need to do is to calculate the coefficient of $\textrm{e}^{\pm irs}$ in $\dbinom{\delta(s)}{\epsilon(s)}$ for $r\in\sigma_{\textrm{B}}(X)=\sigma_{\textrm{B}}(Y)$.

\end{lemma}

\subsection{delta and epsilon Analysis}\vspace{-6mm}

\npgni We concentrate on first order approximations to $\dbinom{\delta}{\epsilon}$, namely\vspace{3mm}

$$\dbinom{\delta}{\epsilon}=\dbinom{\delta_{0}}{\epsilon_{0}}+e\dbinom{\delta_{1}}{\epsilon_{1}},$$

\npgni where $e$ is the orbital eccentricity and\vspace{-3mm}

$$\dbinom{\delta_{1}}{\epsilon_{1}}=A^{-1}\begin{pmatrix}3\omega^2\cos{\omega t}\delta^{(0)} & (4\omega\cos{\omega t}D-2\omega^2\sin{\omega t})\epsilon^{(0)} \\[0.4em] (2\omega^2\sin{\omega t}-4\omega\cos{\omega t}D)\delta^{(0)} & 9\omega^2\cos{\omega t}\epsilon^{(0)} \end{pmatrix},$$

\npgni with $A=\begin{pmatrix}D^2-\frac{9}{4}\omega^2 & 2\omega D+\Omega^2 \\
[0.4em] -2\omega D+\Omega^2 & D^2-\frac{3}{4}\omega^2 \end{pmatrix}$ and $D=\dfrac{d}{dt}$.

\npgni We write $\Delta=\textrm{Det}A$. Evidently $\omega\pm\alpha\in\sigma_{\textrm{B}}\dbinom{\delta_{1}(s)}{\epsilon_{1}(s)}$ so e.g. writing $\delta_{1}^{1}$, $\epsilon_{1}^{1}$ terms in $\textrm{e}^{i(\omega+\alpha)t}$, $\delta_{0}^{1}$, $\epsilon_{0}^{1}$ terms in $\textrm{e}^{i\omega t}$ in $\delta_{0}$, $\epsilon_{0}$\vspace{-3mm}

$$\dbinom{\delta_{1}^{1}}{\epsilon_{1}^{1}}=\Delta^{-1}\begin{pmatrix}D^2-\frac{9}{4}\omega^2 & 2\omega D+\Omega^2 \\[0.4em] -2\omega D+\Omega^2 & D^2-\frac{3}{4}\omega^2 \end{pmatrix}\begin{pmatrix}\frac{3}{2}\omega^2\delta_{0}^{1}+(2\omega i\alpha+i\omega^2)\epsilon_{0}^{1} \\[0.4em] (-2\omega i\alpha-i\omega^2)\delta_{0}^{1}+\frac{9}{2}\omega^2\epsilon_{0}^{1}\end{pmatrix}\textrm{e}^{i(\omega+\alpha)t},$$\vspace{-3mm}

\npgni $D=\dfrac{d}{dt}$, $\Delta=(D^4+\omega^2D^2-\Omega^2+\frac{27}{16}\omega^4)\rceil_{D=(\omega+\alpha)i}$ etc.

\npgni Assume $\delta_{0}(t)=\delta_{0}^{1}\textrm{e}^{i\alpha t}+\delta_{0}^{1}\textrm{e}^{-i\alpha t}+\delta_{0}^{3}\textrm{e}^{i\beta t}+\delta_{0}^{3}\textrm{e}^{-i\beta t}$ etc.

\npgni Contribution of $\delta_{1}^{0}$ to $\delta_{1}^{1}$ reads

\npgni $\Delta^{-1}\left((2\omega(\alpha+\omega)-i\Omega^2)\omega(\omega+2\alpha)-\frac{3}{2}\omega^2\left(\frac{13}{4}\omega^2+2\omega\alpha+\alpha^2\right)\right)\delta_{0}^{1}\textrm{e}^{i(\omega+\alpha)t}$

\npgni \hspace{120mm}$=A_{\delta}^{\delta}(\alpha)\delta_{0}^{1}\textrm{e}^{i(\omega+\alpha)t}$.\vspace{-5mm}

\npgni Contribution of $\epsilon_{0}^{1}$ to $\delta_{1}^{1}$ reads

\npgni $\Delta^{-1}\left(\frac{9}{2}\omega^2(2i\omega(\alpha+\omega)+\Omega^2)-i\omega(\omega+2\alpha)\left(\frac{13}{4}\omega^2+2\omega\alpha+\alpha^2\right)\right)\epsilon_{0}^{1}\textrm{e}^{i(\omega+\alpha)t}$

\npgni \hspace{120mm}$=A_{\epsilon}^{\delta}(\alpha)\epsilon_{0}^{1}\textrm{e}^{i(\omega+\alpha)t}$.\vspace{-5mm}

\npgni Contribution of $\delta_{0}^{1}$ to $\epsilon_{1}^{1}$ reads

\npgni $\Delta^{-1}\left(\frac{3}{2}\omega^2(-2i\omega(\alpha+\omega)+\Omega^2)+\omega(\omega+2\alpha)\left(2\omega(\omega+\alpha)-i\Omega^2\right)\right)\delta_{0}^{1}\textrm{e}^{i(\omega+\alpha)t}$

\npgni \hspace{120mm}$=A_{\delta}^{\epsilon}(\alpha)\delta_{0}^{1}\textrm{e}^{i(\omega+\alpha)t}$.\vspace{-5mm}

\npgni Contribution of $\epsilon_{0}^{1}$ to $\epsilon_{1}^{1}$ reads

\npgni $\Delta^{-1}\left((2\omega(\alpha+\omega)+i\Omega^2)\omega(\omega+2\alpha)+\frac{9}{2}\omega^2\left(2i\omega(\alpha+\omega)+\Omega^2\right)\right)\epsilon_{0}^{1}\textrm{e}^{i(\omega+\alpha)t}$

\npgni \hspace{120mm}$=A_{\epsilon}^{\epsilon}(\alpha)\epsilon_{0}^{1}\textrm{e}^{i(\omega+\alpha)t}$.

\npgni Of course there are similar terms for the $\beta$ root, $-\beta$ root and $-\alpha$ root. For the terms in $\textrm{e}^{i(\omega-\alpha)t}$ contributing we merely set $\delta_{0}^{1}\rightarrow\bar\delta_{0}^{1}$, $\epsilon_{0}^{1}\rightarrow\bar\epsilon_{0}^{1}$ and $\alpha\rightarrow-\alpha$. And for terms in $\textrm{e}^{i(\omega\pm\beta)t}$ we merely replace $\alpha$ by $\beta$. So as an example we obtain,

\npgni $\delta(t)=\delta_{0}(t)+e\left(A_{\delta}^{\delta}(\alpha)\delta_{0}^{1}\textrm{e}^{i(\omega+\alpha)t}+A_{\delta}^{\delta}(-\alpha)\bar\delta_{0}^{1}\textrm{e}^{i(\omega-\alpha)t}+\bar A_{\delta}^{\delta}(\alpha)\bar\delta_{0}^{1}\textrm{e}^{-i(\omega+\alpha)t}\right.$

\npgni \hspace{35mm}$\left.+\bar A_{\delta}^{\delta}(-\alpha)\delta_{0}^{1}\textrm{e}^{-i(\omega-\alpha)t}+A_{\epsilon}^{\delta}(\alpha)\epsilon_{0}^{1}\textrm{e}^{i(\omega+\alpha)t}+A_{\epsilon}^{\delta}(-\alpha)\bar\epsilon_{0}^{1}\textrm{e}^{i(\omega-\alpha)t}\right.$

\npgni \hspace{55mm}$\left.+\bar A_{\epsilon}^{\delta}(\alpha)\bar\epsilon_{0}^{1}\textrm{e}^{-i(\omega+\alpha)t}+\bar A_{\epsilon}^{\delta}(-\alpha)\epsilon_{0}^{1}\textrm{e}^{-i(\omega-\alpha)t}+\beta\;\textrm{terms}\right)$

\npgni $\delta(t)=\delta_{0}(t)+e\textrm{e}^{i(\omega+\alpha)t}\left(A_{\delta}^{\delta}(\alpha)\delta_{0}^{1}+A_{\epsilon}^{\delta}(\alpha)\epsilon_{0}^{1}\right)+\;\textrm{c.c.}$

\npgni \hspace{35mm}$+e\textrm{e}^{i(\omega-\alpha)t}\left(A_{\delta}^{\delta}(-\alpha)\bar\delta_{0}^{1}+A_{\epsilon}^{\delta}(-\alpha)\bar\epsilon_{0}^{1}\right)+\;\textrm{c.c.}+\beta\;\textrm{terms}$.

\npgni We therefore see that there are 6 constants, the original 2 plus $e|A_{\delta}^{\delta}(\alpha)\delta_{0}^{1}+A_{\epsilon}^{\delta}(\alpha)\epsilon_{0}^{1}|$ etc. In addition there are similar terms for $\epsilon(t)$.

\subsection{Constants of Integration in Detail}

\npgni We write $\alpha=f_{1}\omega$ etc. $D=i(\omega+\alpha)$ in cases considered below and

\npgni $\Delta=D^4+\omega^2D^2-\Omega^4+\dfrac{27}{16}\omega^4$, $\Omega^2=\dfrac{3\sqrt{3}}{4}\omega^2$, giving in our case $\Delta=(1+f_{1})^2(f_{1}^2+2f_{1})\omega^4$.

$$A_{\delta}^{\delta}(\alpha)=\frac{\omega^4\left((1+2f_{1})(2(1+f_{1})-i\frac{3\sqrt{3}}{4})-\frac{3}{2}(\frac{13}{4}+2f_{1}+f_{1}^2)\right)}{\omega^4(1+f_{1})^2(f_{1}^2+2f_{1})},$$

\npgni i.e.\hspace{16mm} $A_{\delta}^{\delta}(\alpha)=\dfrac{\left(-\frac{23}{8}+3f_{1}+\frac{5}{2}f_{1}^2-i\frac{3\sqrt{3}}{4}(1+2f_{1})\right)}{(1+f_{1})^2(f_{1}^2+2f_{1})}$.

$$A_{\epsilon}^{\delta}(\alpha)=\frac{\omega^4\left(\frac{9}{2}(\frac{3\sqrt{3}}{4}+2i(1+f_{1}))-i(1+2f_{1})(\frac{13}{4}+2f_{1}+f_{1}^2)\right)}{\omega^4(1+f_{1})^2(f_{1}^2+2f_{1})},$$

\npgni i.e.\hspace{14mm} $A_{\epsilon}^{\delta}(\alpha)=\dfrac{\left(\frac{27\sqrt{3}}{8}+i(\frac{23}{4}+\frac{1}{2}f_{1}-5f_{1}^2-2f_{1}^3)\right)}{(1+f_{1})^2(f_{1}^2+2f_{1})}$.

\npgni So for the above two terms we get contributions to $\delta(t)$

$$\delta(t)=\delta_{0}(t)+e\left(A_{\delta}^{\delta}(\alpha)\delta_{0}^{1}+A_{\epsilon}^{\delta}(\alpha)\epsilon_{0}^{1}\right)\textrm{e}^{i(\omega+\alpha)t}+\;\textrm{c.c.}$$

\npgni So a constant of integration here would have to be the coefficient of $\cos((\omega+\alpha)t+\psi)$, where $\psi=\arg\left(A_{\delta}^{\delta}(\alpha)\delta_{0}^{1}+A_{\epsilon}^{\delta}(\alpha)\epsilon_{0}^{1}\right)$, which reads:-

$$2e|A_{\delta}^{\delta}(\alpha)\delta_{0}^{1}+A_{\epsilon}^{\delta}(\alpha)\epsilon_{0}^{1}|\cos((\omega+\alpha)t+\psi).$$

\npgni $\delta_{0}^{1}$ and $\epsilon_{0}^{1}$ are 'constants of motion' for the zero eccentricity circular orbit and so because of the factor of $e$ outside, correct to first order in $e$, they are constant in this expansion. Similar reasoning applies to the other terms in this sum for $\delta(t)$ and $\epsilon(t)$. Incidentally from Szebehely (Ref.[30]), for the Sun-Jupiter system, $f_{1}=0.996758$ and $f_{2}=0.080464$.

\npgni In the case of the circular orbit there exist periodic elliptical orbit solutions for the third body (asteroid):

$$X=4\omega\alpha C\sin \alpha t\spc;\spc Y=2(\alpha^2+2\omega_{X}^2)C\cos \alpha t.$$

\npgni These yield

$$\delta_{0}^{1}=\left(\frac{(\alpha^2+2\omega_{X}^2)}{2}\sin{\gamma}+i2\omega\alpha\cos{\gamma}\right)C_{1},$$

$$\epsilon_{0}^{1}=\left(-\frac{(\alpha^2+2\omega_{X}^2)}{2}\cos{\gamma}+i2\omega\alpha\sin{\gamma}\right)C_{1}$$

\npgni and the constant of integration term associated with $\textrm{e}^{i(\omega+\alpha)t}$ in $(\delta(t),\epsilon(t))$ is

$$A_{\delta}^{\delta}(\alpha)=e\Bigg|\dfrac{\left(-\frac{23}{8}+3f_{1}+\frac{5}{2}f_{1}^2-i\frac{3\sqrt{3}}{4}(1+2f_{1})\right)\delta_{0}^{1}+\left(\frac{27\sqrt{3}}{8}+\frac{1}{2}f_{1}-5f_{1}^2-2f_{1}^3)\right)\epsilon_{0}^{1}}{(1+f_{1})^2(f_{1}^2+2f_{1})}\Bigg||C_{1}|.$$

\npgni The first factor is of course a constant as $|C_{1}|$ is a constant in the circular orbit case and because of the external factor $e$, to within first order, we have a constant of integration for our solution even for $e\ne0$, $e$ small.

\npgni \textbf{Part 2}

\section{Semi-classical and Quantum Mechanical Results for WIMPs}

\subsection{Asymptotics of Newtonian Quantum Gravity}\vspace{-5mm}

\npgni We envisage a cloud of diffusing quantum particles (WIMPs), condensing to form a celestial body orbiting on a curve $C_{0}$, where the WIMPs' trajectories converge to classical periodic orbits on $C_{0}$, eventually fusing together with other particles to form the classical body. In our example the celestial body is the Trojan asteroid moving on the orbit, centred at $\textrm{L}=\mathscr{L}_{4,5}$, the Lagrange equilibrium points for the restricted 3-body problem. The convergence to the classical orbit (in the central manifold for this linearised problem) is due to our diffusion process which is the semi-classical limit of Nelson's stochastic mechanics for a Schrödinger stationary state wave function $\psi$, where $\psi\sim\textrm{exp}\left(\frac{R+iS}{\hbar}\right)$ as $\hbar\sim0$. (See Refs.[9],[10],[11]). In our case it turns out that the relevant potential energy is the isotropic harmonic oscillator as we shall see.

\npgni Here the quantum particle density corresponding to $\psi$, $\rho\sim\textrm{exp}\left(\dfrac{2R}{\hbar}\right)$ as $\hbar\sim0$, $\rho$ (after suitable normalisation) being the invariant density for the semi-classical limit of our diffusion process. If the quantum Hamiltonian,

$$H=\frac{{\vct p}^2}{2}+V({\vct q}),\;\;\;H\psi=-\frac{\hbar^2}{2}\Delta\psi+V\psi=E\psi,\;\;\;E\;\textrm{being the energy},$$

\npgni $V$ the potential energy for the gravitational forces involved and $\psi=\psi(\vct x)$, $\vct x$, a point in the configuration space. We obtain as $\hbar\sim0$

$$-\left(\frac{\hbar}{2}(\Delta R+i\Delta S)+\frac{(\boldsymbol{\nabla}R+i\boldsymbol{\nabla}S)^2}{2}\right)\psi+V\psi=E\psi,$$

\npgni so it is necessary on the support of $\psi$ that in the limit:-

$$\boldsymbol\nabla R.\boldsymbol\nabla S=0,\spc2^{-1}(|\boldsymbol\nabla S|^2-|\boldsymbol\nabla R|^2)+V=E.$$

\npgni These are our semi-classical equations valid in a neighbourhood of $C_{0}$. The corresponding semi-classical dynamics for a particle in the cloud is

$$\frac{d}{dt}\vct{X}_t=\vct{b}(\vct{X}_t)=(\boldsymbol\nabla S+\boldsymbol\nabla R)(\vct{X}_t),$$

\npgni $R$ and $S$ satisfying the above equations. We assume $R$ achieves its global maximum on the curve $C_{0}$. Since

$$\frac{d}{dt}R(\vct{X}_t)=((\boldsymbol\nabla S+\boldsymbol\nabla R).\boldsymbol\nabla R)(\vct{X}_t)=|\boldsymbol\nabla R|^2(\vct{X}_t)\ge 0$$

\npgni $R(\vct{X}_t)\nearrow R_{\textrm{max}}$ as $t\nearrow\infty$ giving us the desired convergence to orbits on $C_{0}$, where $\boldsymbol\nabla R=\vct{0}$ and

$$2^{-1}|\boldsymbol\nabla S|^2+V=E$$

\npgni i.e. $S\rceil_{C_0}=S_0$, the classical Hamilton-Jacobi function corresponding to classical mechanics on $C_{0}$ in potential $V$. Moreover, we see that

$$2^{-1}(|\boldsymbol\nabla S|^2+|\boldsymbol\nabla R|^2)=E-V+|\boldsymbol\nabla R|^2,$$

\npgni so the asymptotics of Newtonian quantum gravity at lowest level are governed by the effective potential

$$V_{\textrm{eff}}=V-|\boldsymbol\nabla R|^2.$$

\npgni (See Refs.[9],[10],[11]).

\npgni If we know the Schrödinger wave function concentrated with minimal uncertainty on the curve of the classical orbit $C_{0}$ and can compute its asymptotics we will have detailed knowledge of $V_{\textrm{eff}}$. We have computed these asymptotics for the main potentials arising for Newtonian gravity i.e. for the Kepler problem and isotropic oscillator giving spiral orbits converging to well known elliptical orbits. (see Ref.[32]). In the absence of the solution to the Schrödinger equation concentrated on $C_{0}$ the question is how to find $R$ and $S$ in a neighbourhood of $C_{0}$. We discuss this problem in 2-dimensions next.

\subsection{Approximate Solutions in a Neighbourhood of the Classical Orbit in 2-
Dimensions}

\npgni We consider semi-classical motion in the effective potential $(V-|\boldsymbol\nabla R|^2)$. Since the particle velocity in a neighbourhood of $C_{0}$ is

$$\vct{v}=\boldsymbol\nabla R+\boldsymbol\nabla S,\;\;\;v=|\vct{v}|=|\boldsymbol\nabla (R+S)|$$

\npgni i.e. $v=\sqrt{|\boldsymbol\nabla R|^2+|\boldsymbol\nabla S|^2}$, $v$ is the speed at time $t$ if we evaluate the right hand side at $\vct{X}_t$. Initially, we work in 2-dimensions only.

\npgni The equations of motion reduce to

$$\kappa v^2=\pm(\vct{n}.\boldsymbol\nabla)(V-|\boldsymbol\nabla R|^2)\;\;\;\;\;(i),$$

$$\frac{1}{v}\frac{d}{dt}\left(\frac{v^2}{2}\right)=-(\vct{t}.\boldsymbol\nabla)(V-|\boldsymbol\nabla R|^2)\;\;\;\;\;(ii)$$

\npgni $\kappa$ being the curvature of the semi-classical orbit at time $t$, $\vct{n}$ and $\vct{t}$ unit normals and tangents, at time $t$, to the orbit. Evidently $(ii)$ is a simple consequence of the energy conservation equation so there is only one remaining equation in 2-dimensions in the neighbourhood of $C_{0}$, namely (NLE),

$$\kappa(|\boldsymbol\nabla R|^2+|\boldsymbol\nabla S|^2)=\left(1+\frac{|\boldsymbol\nabla R|^2}{|\boldsymbol\nabla S|^2}\right)^{-\frac{1}{2}}\left(\frac{|\boldsymbol\nabla R|}{|\boldsymbol\nabla S|}\widehat{\boldsymbol\nabla S}-\widehat{\boldsymbol\nabla R}\right).\boldsymbol\nabla(V-|\boldsymbol\nabla R|^2)\;\;\;\;(\textrm{NLE}),$$

$$\boldsymbol\nabla R=\left(\dfrac{\partial R}{\partial x},\dfrac{\partial R}{\partial y}\right),\;\;\;\widehat{\boldsymbol\nabla S}=\dfrac{1}{|\boldsymbol\nabla R|}\left(-\dfrac{\partial R}{\partial y},\dfrac{\partial R}{\partial x}\right)$$

\npgni in cartesians and $|\boldsymbol\nabla S|=\sqrt{2(E-V)+|\boldsymbol\nabla R|^2}$.

\npgni Parallel Curves

\npgni Parallel curves are curves at a fixed distance $|d|$ away from an existing curve, say $C_{0}$, with parametric equations $(x,y)=(x_{0}(t),y_{0}(t))$. The parallel curve parametric equations are given by $(x,y)=(x_{d}(t),y_{d}(t))$, where

$$x_{d}(t)=x_{0}(t)+\frac{d\dot{y}_{0}(t)}{\sqrt{\dot{x}_{0}^2(t)+\dot{y}_{0}^2(t)}},\;\;\;y_{d}(t)=y_{0}(t)-\frac{d\dot{x}_{0}(t)}{\sqrt{\dot{x}_{0}^2(t)+\dot{y}_{0}^2(t)}}.$$

\npgni $d$ can be positive or negative and the distance between the two curves is $|d|$.

\npgni It is easy to show that $(\dot{x}_{d}(t),\dot{y}_{d}(t))$ is parallel to $(\dot{x}_{0}
(t),\dot{y}_{0}(t))$ and for curvatures $\kappa_{0}$ and $\kappa_{d}$,

$$\kappa_{d}=\frac{\kappa_{0}}{1+d\kappa_{0}}.$$

\npgni So in our neighbourhood as $d\sim0$, for $\vct{n}$ the inward pointing normal,

$$\kappa_{d}(|\boldsymbol\nabla R|^2+|\boldsymbol\nabla S|^2)=\pm\vct{n}.\boldsymbol\nabla(V-|\boldsymbol\nabla R|^2),$$

\npgni where $|\boldsymbol\nabla S|^2\rightarrow|2(E-V_{0})|$ as $d\sim0$. We therefore obtain:

\begin{lemma}

\npgni If at time $t$, $\vct{X}_t=(x_{d}(t),y_{d}(t))$ for some $(x_{0}(t),y_{0}(t))\in C_{0}$ then as $t\nearrow\infty$ and $|d|\sim0$, we obtain for $\dfrac{\partial}{\partial n}$, the outward normal derivative on $C_{0}$,

$$|\boldsymbol\nabla R|^2\sim d\left(2(E-V_{0})\kappa_{0}+\dfrac{\partial V_{0}}{\partial n}\right)$$

\npgni and

$$|\boldsymbol\nabla S|^2\sim d\left(2(E-V_{0})\kappa_{0}+\dfrac{\partial V_{0}}{\partial n}\right)+2(E-V_{0}),$$

\npgni where the subscript $0$ means that we have to evaluate the expressions at $(x_{0}(t),y_{0}(t))$.

\end{lemma}

\npgni To push this 2-dimensional result any further one needs the full complexity of the equation (NLE). In 3-dimensions there are similar formulae involving the torsion of the path as well as its curvature at time $t$. In the next section we give the full solution to this puzzle for the Trojan asteroid problem in 3-dimensions. This exploits the dynamical symmetry group of the Hamiltonian for the isotropic harmonic oscillator achieving minimal uncertainty concentration on $C_{0}$. (See Ref.[27]).

\npgni \textbf{Remarks}

\npgni \textbf{1.} The first order quantum effect is to reduce the kinetic energy of the particle by $|\boldsymbol\nabla R|^2$ both inside and outside $C_{0}$. As for the $\dfrac{v^2}{\rho_{0}}$ term, where $\rho_{0}$ is the radius of curvature of the semi-classical orbit, since $\dfrac{\partial}{\partial n}|\boldsymbol\nabla R|^2<0$ outside $C_{0}$ and $\dfrac{\partial}{\partial n}|\boldsymbol\nabla R|^2>0$ inside $C_{0}$, we\vspace{-3mm}

\npgni see that outside $C_{0}$ the radius of curvature has to decrease compared with the 
classical value.

\npgni \textbf{2.} $|\boldsymbol\nabla R|$ gives the deviation from the parallel curve $C_{d}$, a distance $d$ away from $C_{0}$ since in our neighbourhood $\vct{v}\sim\boldsymbol\nabla S+|\boldsymbol\nabla R|\widehat{\boldsymbol\nabla R}$, the last term tending to zero as $t\nearrow\infty$.

\npgni It would be good to solve the (NLE) for small $|\boldsymbol\nabla R|$ but we postpone that for now. Instead we recall exact results for the two most important potentials $V$ in Newtonian gravity and detail the physical effects in the next section of this paper.

\npgni Of course the above equations $(i)$ and $(ii)$ recapitulate Ed Nelson's result that the Schrödinger equation for $\psi=\textrm{exp}(R+iS)$ linearises the horribly non-linear equation in (R,S) required for Newton's 2\textsuperscript{nd} Law to be valid in stochastic mechanics (Ref.[24]). This prompts us to return to exact solutions of the Schrödinger equation and the asymptotics to derive the effects of quantum gravity on the curvature of semi-classical orbits for the Kepler problem and for the isotropic harmonic oscillator. We believe the former has applications to galactic evolution and the latter to the formation of, for example, Trojan asteroids.

\npgni It remains for us to prove that Newton's 2\textsuperscript{nd} Law of motion is valid in our set-up in 2 and 3 dimensions guaranteeing the solution of NLE.

\begin{lemma}

$$((\boldsymbol\nabla R+\boldsymbol\nabla S).\boldsymbol\nabla)(\boldsymbol\nabla R+\boldsymbol\nabla S)=-\boldsymbol\nabla(V-|\boldsymbol\nabla R|^2).$$

\end{lemma}

\begin{proof}

$$\textrm{l.h.s.}=\boldsymbol\nabla\left(\frac{(\boldsymbol\nabla(R+S))^2}{2}\right).$$

\npgni By energy conservation $2^{-1}(|\boldsymbol\nabla S|^2+|\boldsymbol\nabla R|^2)=E-V+|\boldsymbol\nabla R|^2$ and the result follows from elementary vector algebra.

\end{proof}\vspace{-5mm}

\npgni In the next section we compute $\boldsymbol\nabla R$ and $\boldsymbol\nabla S$ for the isotropic harmonic oscillator potential.

\section{Quantisation and Semi-Classical Mechanics for Isotropic Harmonic Oscillators in 2 and 3-Dimensions}

\subsection{Quantum Connections with Isotropic Harmonic Oscillators}\vspace{-5mm}

\npgni Returning to the Trojan asteroid problem we need to find the state of the cloud of WIMPs so it will condense onto the classical periodic orbits in the central manifold for this linearised problem. From the form of the general solution it follows that, since $D_{2}\propto|C_{3}|$, for small $|D_{2}|$ the corresponding orbits are near the orbit with $D_{2}=0$. Following Dirac (Ref.[7]), further we can think of $D_{2}=0$ as a constraint and, since the Poisson bracket $\{D_{1},D_{2}\}=0$ here, there are no secondary constraints. This follows if either one of the $D'\textrm{s}$ is zero in which case by inspection of the general solution,

$$\ddot{X}=-\tilde\omega^2,\;\;\;\ddot{Y}=-\tilde\omega^2Y,\;\;\;\tilde\omega=\alpha\;\textrm{or}\;\beta,\;\;\;\lambda=\pm i\alpha,\pm i\beta$$\vspace{-10mm}

\npgni being roots of\vspace{-5mm}

$$\lambda^4+\omega^2\lambda^2+4\omega_{X}^2\omega_{Y}^2=0,$$

\npgni with $\omega_{X}^2+\omega_{Y}^2=\dfrac{3}{2}\omega^2, \;\;\omega_{X}^2\omega_{Y}^2=\dfrac{27\mu_{1}\mu_{2}}{16(\mu_{1}+\mu_{2})^2}\omega^4,\;\;{\omega}^{2}=\dfrac{(\mu_{1}+\mu_{2})}{|\textrm{SJ}|^{3}},\;\textrm{and}\;\mu_{1},\mu_{2}$ are

\npgni the gravitational masses of the Sun and Jupiter, respectively, for the Trojan problem. Hence, we have a 2-dimensional isotropic harmonic oscillator such as is associated with a homogeneous cloud of gravitating WIMPs. We also have a similar result in 3-dimensions where the next theorem is relevant:

\begin{theorem}

The isotropic harmonic oscillator elliptic state is, up to normalisation, for $\tilde\lambda=n\hbar$,

$$\psi_n^{\textnormal{HO}}(\vct r)=\textrm{exp}\left(-\frac{\tilde\omega r^2}{2\tilde\lambda}\right)H_n(\sqrt{n}\tilde u),\;\;\vct r=(x,y,z),$$

\npgni where $r^2=x^2+y^2+z^2$, $\tilde u=\sqrt{\dfrac{\tilde\omega}{\tilde\lambda}\left((1-\alpha)\dfrac{x^2}{2}+(1+\alpha)\dfrac{y^2}{2}-i\beta xy\right)}$, $\alpha=\dfrac{1}{e}$,\vspace{-5mm} 

\npgni $\beta=\dfrac{\sqrt{1-e^2}}{e}$, $0<e<1$, $e$ being the eccentricity for our ellipse with equations

$$\frac{x^2}{a^2}+\frac{y^2}{b^2}=1,\spc z=0,$$

\npgni and $H_n$ being a Hermite polynomial. 

\npgni For the quantum Hamiltonian $H=2^{-1}\vct P^2+2^{-1}\tilde\omega^2\vct Q^2$ in 3-dimensions,

$$H\psi_n^{\textnormal{HO}}=\tilde\omega\left(\tilde\lambda+\frac{3}{2}\hbar\right)\psi_n^{\textnormal{HO}},\;\;\;\tilde\lambda=n\hbar\rightarrow E,\;\textrm{the energy}\;\;E=\frac{1}{2}\tilde\omega^2(a^2+b^2).$$

\npgni (For the Trojan asteroid problem, we have to set $a=2\tilde\omega\alpha|C|$, $b=(\alpha^2-2\omega_{X}^2)|C|$, $C=|C_{1}|\;\textrm{or}\;|C_{3}|$ depending upon which of $D_{1}$ or $D_{2}$ is set to zero).

\end{theorem}

\begin{proof}
First $\psi_1(z)=\textrm{exp}\left(-\dfrac{\tilde\omega z^2}{2\tilde\lambda}\right)$ is the ground state of the harmonic oscillator in 1-dimension and if

$$\psi_2(x,y)=\textrm{exp}\left(-\frac{\tilde\omega(x^2+y^2)}{2\tilde\lambda}\right)H_n(\sqrt{n}\tilde u),$$

$$\boldsymbol\nabla\psi_1.\boldsymbol\nabla\psi_2=0, \;\;\textrm{so}\;\; \Delta(\psi_1\psi_2)=\psi_1(\Delta\psi_2)+(\Delta\psi_1)\psi_2.$$

\npgni Further $\Delta\tilde u=0$, so $\Delta H_n(\sqrt{n}\tilde u)=nH_n''(\sqrt{n}\tilde u)|\boldsymbol\nabla\tilde u|^2$. Also,

\begin{flushleft}
\;\;$\Delta\psi_2 = \Delta\left(\textrm{exp}\left(-\dfrac{\tilde\omega(x^2+y^2)}{2\tilde\lambda}\right)H_n(\sqrt{n}\tilde u)\right)$
\end{flushleft}

\begin{flushleft}
$\spc\;\;=\textrm{exp}\left(-\dfrac{\tilde\omega(x^2+y^2)}{2\tilde\lambda}\right)\Delta H_n(\sqrt{n}\tilde u)+H_n(\sqrt{n}\tilde u)\Delta\textrm{exp}\left(-\dfrac{\tilde\omega(x^2+y^2)}{2\tilde\lambda}\right)$
\end{flushleft}

\begin{flushleft}
$\spc\spc\spc\spc\spc\spc\spc+2\boldsymbol\nabla\left(\textrm{exp}\left(-\dfrac{\tilde\omega(x^2+y^2)}{2\tilde\lambda}\right)\right).\boldsymbol\nabla H_n(\sqrt{n}\tilde u).$
\end{flushleft}

\npgni A computation using the properties of Hermite polynomials yields the result

$$-\frac{\hbar^2}{2}\Delta\psi_n^{\textrm{HO}}+\frac{\tilde\omega^2}{2}(x^2+y^2+z^2)\psi_n^{\textrm{HO}}=\tilde\omega\left(\tilde\lambda+\frac{3}{2}\hbar\right)\psi_n^{\textrm{HO}},\spc\tilde\lambda=n\hbar.$$
\end{proof}

\begin{lemma}
$$\lim_{n\nearrow\infty}\frac{H'_{n}(\sqrt{n}u)}{\sqrt{n}H_{n}(\sqrt{n}u)}=u-\sqrt{u^2-2}.$$
\end{lemma}

\begin{proof}

\npgni Let $Q_{n}=\dfrac{H'_{n}(\sqrt{n}u)}{\sqrt{n}H_{n}(\sqrt{n}u)}$. Clearly $Q_{n}=\dfrac{2\sqrt{n}H_{n-1}(\sqrt{n}u)}{H_{n}(\sqrt{n}u)}$. Then using the

\npgni standard recurrence relation

$$H_{n+1}(x)=2xH_{n}(x)-2nH_{n-1}(x),$$

\npgni with $x=\sqrt{n}u$ we can conclude

$$\dfrac{2\sqrt{n+1}}{\sqrt{n}Q_{n+1}}=2u-Q_{n}.$$

\npgni Assuming that $Q_{n}\rightarrow{q}$ as $n\rightarrow{\infty}$ then,

$$\dfrac{2}{q}=2u-q,$$\vspace{-8mm}

\npgni giving the desired result.

\end{proof}

\npgni For application to the Trojan asteroid problem we have to set $\tilde\omega=\alpha\;\textrm{or}\;\beta$, where

$$\alpha=\dfrac{1}{\sqrt{2}}\left(1+\sqrt{1-\dfrac{27\mu_{1}\mu_{2}}{(\mu_{1}+\mu_{2})^2}}\right)^{\frac{1}{2}}\omega\;\;;\;\;\beta=\dfrac{1}{\sqrt{2}}\left(1-\sqrt{1-\dfrac{27\mu_{1}\mu_{2}}{(\mu_{1}+\mu_{2})^2}}\right)^{\frac{1}{2}}\omega.$$

\npgni Setting $z=0$ gives the asymptotics of $\psi$ for the isotropic harmonic oscillator potential,

\npgni $V=\dfrac{\tilde\omega^2}{2}(x^2+y^2)$, for the elliptical orbit, $\psi\sim\textrm{exp}\left(\frac{\tilde R+i\tilde S}{\hbar}\right)$ as $\hbar\sim0$, $\vct{r}=(x,y)$

$$(\tilde R+i\tilde S)(\vct{r})=-\frac{1}{2}\tilde\omega r^2+\frac{\tilde\lambda}{2}\tilde u^2\left(1-\sqrt{1-\frac{2}{\tilde u^2}}\right)+\tilde\lambda\left(\tilde u+\sqrt{\tilde u^2-2}\right),$$

\npgni with $\tilde u=\sqrt{\dfrac{\tilde\omega}{\tilde\lambda}\left((1-\alpha)\dfrac{x^2}{2}+(1+\alpha)\dfrac{y^2}{2}-i\beta xy\right)}$, $\alpha=\dfrac{1}{e}$ and $\alpha^2-\beta^2=1$.

\npgni This is the analogue of the 2-dimensional Kepler problem asymptotics for Keplerian

\npgni ellipses, $\psi\sim\textrm{exp}\left(\frac{R+iS}{\hbar}\right)$ for $V=-\dfrac{\mu}{x^2+y^2}$, $\vct{r}=(x,y)$,

$$R+iS=-\frac{\mu}{\lambda}r+\frac{\lambda \nu}{2}\left(1-\sqrt{1-\frac{4}{\nu}}\right)-\lambda\ln \nu-2\lambda\ln\left(1-\sqrt{1-\frac{4}{\nu}}\right),$$

\npgni where $\nu=\dfrac{\mu}{\lambda^2}\left(r-\dfrac{x}{e}-\dfrac{iy\sqrt{1-e^2}}{e}\right)$ and $\lambda=\sqrt{a\mu}$, $a$
being the semi-major axis of the ellipse.

\npgni In what follows $\psi_{n,e}$ will denote the Keplerian analogue of $\psi_n^{\textrm{HO}}$. The corresponding orbits here spiral into the Keplerian ellipses.

\subsection{Quantum Corrections, Curvature and Torsion for Trojan WIMPs}\vspace{-5mm}

\npgni If no vector potentials are involved all one has to do to compute the quantum correction to curvature of the spirals is to replace $V$ by $V_{\textrm{eff}}$, where $V_{\textrm{eff}}=V-|\boldsymbol\nabla R|^2$, for which $|\boldsymbol\nabla R|^2$ can be read off from the above, in the formula for classical curvature $\kappa_{c}$, giving\vspace{5mm}

$$\kappa_{q}=\frac{\Big|\dfrac{\partial}{\partial n}(E-V+|\boldsymbol\nabla R|^2)\Big|}{2(E-V+|\boldsymbol\nabla R|^2)},$$

\npgni where the normal derivative to the orbit is\vspace{3mm}

$$\dfrac{\partial}{\partial n}=\sqrt{\left(1+\dfrac{|\boldsymbol\nabla R|^2}{|\boldsymbol\nabla S|^2}\right)}\left(\dfrac{|\boldsymbol\nabla R|}{|\boldsymbol\nabla S|}\widehat{\boldsymbol\nabla S}-\widehat{\boldsymbol\nabla R}\right).\boldsymbol\nabla,$$

\npgni and in 2-dimensions,

$$\boldsymbol\nabla S=\sqrt{1+\dfrac{2(E-V)}{|\boldsymbol\nabla R|^2}}\left(-\dfrac{\partial R}{\partial y},\dfrac{\partial R}{\partial x}\right),$$

\npgni where $|\boldsymbol\nabla R|^2=0$ in the classical case, $S$ being the Hamilton-Jacobi function.

\npgni Unfortunately the linearised 3-body problem is slightly more complicated.

\npgni Lemma 4.2 is not appropriate to the linearised Trojan asteroid problem as in this case we have used rotating axes so the dynamics involves a Coriolis force. To generalise the above we need to include a vector potential e.g. $\vct{A}=\omega(-y,x,0)$ in the above. From the corresponding Schrödinger wave function for a stationary state $\psi\sim\textrm{exp}\left(\frac{R+iS}{\hbar}\right)$, with energy $E$, for scalar and vector potentials $V$ and $\vct{A}$, respectively, the corresponding semi-classical mechanics 
reduces to :-

$$(\boldsymbol\nabla S-\vct{A}).\boldsymbol\nabla R=0,\spc2^{-1}(|\boldsymbol\nabla S-\vct{A}|^2-|\boldsymbol\nabla R|^2)+V=E,$$

\npgni where the dynamics corresponds to

$$\dot{\vct{X}}_{t}=\frac{d\vct{X}_{t}}{dt}=(\boldsymbol\nabla(R+ S)-\vct{A})(\vct{X}_{t}),\;\;t\ge0\;\;\textrm{and}\;\;2^{-1}\dot{\vct{X}}^{2}_{t}+V_{\textrm{eff}}=E,\;V_{\textrm{eff}}=V-|\boldsymbol\nabla R|^2.$$

\npgni Although the above equations are non-linear there is a linear superposition principle inherited from the Schrödinger equation as in the case $\vct{A}=\boldsymbol{0}$ (see Ref.[32]). The point is that here the entropy is still $(-R)$.

\npgni As you will see there is another way of confirming the validity of Hamilton's equations in this problem for $\vct{A}=\omega(-y,x,0)$ and $\textrm{curl}\vct{A}=(0,0,2\omega)$.

\npgni \textbf{Example}

\npgni For a unit positively charged particle P in the constant magnetic field, $\vct{B}=(0,0,B)$, $\vct{A}=2^{-1}\vct{r}\times\vct{B}$, in 2-dimensions, there is a circular spiral orbit of radius $a_{0}$, with $\overset{\rightharpoonup}{\textrm{OP}}=\vct{X}_{t}$, with $\vct{X}=(x,y,0)$,

$$\boldsymbol{\nabla}R=\dfrac{1}{r^2}\sqrt{\left(\dfrac{B}{2}r^2-L\right)^2-2Er^2}\;(x,y),\;\;\;\boldsymbol{\nabla}S=L\left(-\dfrac{y}{r^2},\dfrac{x}{r^2}\right),\;\;\;r=\sqrt{x^2+y^2}.$$

\npgni $\dot{\vct{X}}_{t}=(\boldsymbol\nabla(R+ S)-\vct{A})(\vct{X}_{t})$ \;so that for \;$t\ge0$, \;$|r_{t}^2-a_{0}^2|=|r_{0}^2-a_{0}^2|\textrm{e}^{-Bt}$, where

\npgni $a_{0}^2=\dfrac{2E}{B^2}=-\dfrac{2L}{B}$.\vspace{5mm}

\begin{lemma}
For the above semi-classical orbit, $\vct{v}=\boldsymbol\nabla(R+S)-\vct{A}$, $\vct{v}$ the particle velocity and $\vct{v}=\dot{\vct{X}}_{t}$, whilst the acceleration,

$$\ddot{\vct{X}}_{t}=-\boldsymbol\nabla(V-|\boldsymbol\nabla R|^2)+\dot{\vct{X}}_{t}\times\textrm{curl}\vct{A},$$

\npgni the last term being a Coriolis force and the r.h.s. being evaluated at $\vct{X}_{t}$.
\end{lemma}

\begin{proof}
The chain rule gives, $\ddot{\vct{X}}_{t}=\dfrac{d^2\vct{X}_{t}}{dt^2}=(\dot{\vct{X}}_{t}.\boldsymbol\nabla)\dot{\vct{X}}_{t}$. Since for any vector $\vct{a}$,

\npgni $\boldsymbol\nabla(2^{-1}\vct{a}^2)=(\vct{a}.\boldsymbol\nabla)\vct{a}+\vct{a}\times\textrm{curl}\vct{a}$, setting $\vct{a}=\boldsymbol\nabla(R+ S)-\vct{A}$, the r.h.s. being evaluated

\npgni at $\vct{X}_{t}$, $\vct{a}=\vct{v}$,

$$\ddot{\vct{X}}_{t}=\boldsymbol\nabla(2^{-1}\vct{v}^2+V-|\boldsymbol\nabla R|^2)-\boldsymbol\nabla(V-|\boldsymbol\nabla R|^2)+\dot{\vct{X}}_{t}\times\textrm{curl}\vct{A},$$

\npgni giving from energy conservation,

$$\ddot{\vct{X}}_{t}=-\boldsymbol\nabla(V-|\boldsymbol\nabla R|^2)+\dot{\vct{X}}_{t}\times\textrm{curl}\vct{A}.$$

\end{proof}\vspace{-5mm}

\npgni So we obtain in 3-dimensions, $v=|\vct{v}|$ and

$$\kappa v^2\vct{n}+\frac{1}{v}\frac{d}{dt}\left(\frac{v^2}{2}\right)\vct{t}=-\boldsymbol\nabla(V-|\boldsymbol\nabla R|^2)+\dot{\vct{X}}_{t}\times\textrm{curl}\vct{A},$$

\npgni $\kappa$ being the curvature of the orbit at $\vct{X}_{t}$. When the motion takes place in the $(x,y)$ plane and $\textrm{curl}\vct{A}$ is perpendicular to the plane, $\textrm{curl}\vct{A}=(0,0,2\omega)$,

$$\kappa v^2=-\vct{n}.\boldsymbol\nabla(V-|\boldsymbol\nabla R|^2)+2\omega v,$$

\npgni where $v=|\boldsymbol\nabla(R+S)-\vct{A}|=|{\dot{X}}_{t}|$ and $V$ is the gravitational potential.

\npgni The last result highlights the importance of $|\boldsymbol\nabla R|^2$ as well as the Coriolis force in calculating the curvature of the semi-classical trajectories in rotating frames. Anyway these results are important in predicting the past of Trojan asteroid systems.

\npgni We conclude here with the full expansion for $|\boldsymbol\nabla R|^2$ for the Coulomb potential in 2-dimensions:

\npgni \hspace{5mm}$|\boldsymbol\nabla R|^2=\dfrac{\mu^2}{4\lambda^2}\left((1+\gamma_{\textrm{R}})^2+\left(\dfrac{(1-\gamma_{\textrm{R}})^2}{e^2}+\dfrac{\gamma_{\textrm{I}}^2(1-e^2)}{e^2}\right)+\dfrac{2(1-\gamma_{\textrm{R}}^2)x}{er}\right.$

\npgni \hspace{100mm}$+\left.\dfrac{2\gamma_{\textrm{I}}(1+\gamma_{\textrm{R}})y\sqrt{1-e^2}}{er}\right)$,\vspace{-5mm}

\npgni where\vspace{-5mm}

$$\gamma_{\textrm{R}}=\left(\frac{1}{2}\sqrt{\frac{\left(er-x-\frac{4\lambda^2e}{\mu}\right)^2+(1-e^2)y^2}{(er-x)^2+(1-e^2)y^2}}+\frac{1}{2}\frac{\left(er-x-\frac{2\lambda^2e}{\mu}\right)^2+(1-e^2)y^2-\frac{4\lambda^4e^2}{\mu^2}}{(er-x)^2+(1-e^2)y^2}\right)\tph,$$

$$\gamma_{\textrm{I}}=-\frac{2\lambda^2e\sqrt{1-e^2}y}{\mu((er-x)^2+(1-e^2)y^2)\gamma_{\textrm{R}}},\;\;\;\lambda^2=\mu a.$$

\npgni There is a similar result for the isotropic harmonic oscillator potential which is easily derived from the results in section 5.\vspace{-3mm}

\npgni It is also possible to calculate semi-classical curvature and torsion in 3-dimensions in specific cases. In particular for circular Keplerian spirals one can compute the quantum curvature and quantum torsion in 3-dimensions of the orbit $\vct{X}_t^0=(x,y,z)$, viz.

$$\kappa_{q}=\pm\frac{a^2(2(2z^2r^2+\rho^4)z^2\rho^2+(2ar^3-\rho^4)^2)^\frac{1}{2}}{(2a^2r^2-2a\rho^2r+\rho^4+\rho^2z^2)^\frac{3}{2}\rho},$$

$$\tau_{q}=\frac{(3\rho^8-a(\rho^4+z^4)r^3)z}{2(2z^2r^2+\rho^4)z^2\rho^2r^2+(2ar^3-\rho^4)^2r^2},$$

\npgni where $r^2=x^2+y^2+z^2$, $\rho^2=x^2+y^2$ and $a=\dfrac{\lambda^2}{\mu}$ is the radius of the corresponding classical circular orbit (see Ref.[12]).

\npgni Further, for the stationary state, $\psi_{n,e}$, the Hamiltonian, $H=2^{-1}\vct{P}^{2}-\mu|\vct{Q}|^{-1}$, is a constant, $H=E$, $E<0$ being the energy. Moreover, if $\tilde{\vct{A}}=(-2E)^{-1/2}\vct{A}$, $\vct{A}$ being the Hamilton–Lenz–Runge vector and $\vct{L}$ is the orbital angular momentum, $\tilde{\vct{A}}$ and $\vct{L}$ are quantum constants of the motion generating the dynamical symmetry group SO(4). Setting the Bohr correspondence limits equal to $\tilde{\vct{a}}=(\tilde{a_{1}},\tilde{a_{2}},\tilde{a_{3}})$, $\boldsymbol \ell=(\ell_{1},\ell_{2},\ell_{3})$, of $\tilde{\vct{A}}$ and $\vct{L}$ defined by the Bohr limits of cartesian coordinates,

$$\tilde{a_{i}}=\lim \psi_{n,e}^{-1} \tilde{A_{i}} \psi_{n,e}, \spc \ell_{i}=\lim \psi_{n,e}^{-1} L_{i} \psi_{n,e}, \spc i=1,2,3,$$

\npgni where $\psi_{n,e}\ne 0$, for $\vct{Z}=\lim(-i\hbar\boldsymbol \nabla\ln\psi_{n,e})$, $\vct{Z}=-i\boldsymbol\nabla R+\boldsymbol\nabla S$, we obtain $\boldsymbol \ell=\vct{r}\wedge \vct{Z}(\vct{r})$, $\vct{a}=\vct{Z}\wedge(\vct{Z}\wedge\vct{r})-\mu r^{-1}\vct{r}$, the semi-classical variables inheriting Pauli’s identities for $\vct{L}$ and $\vct{A}$. So, defining $\tilde{\vct{a}}=(-2E)^{-1/2}\vct{a}$, we have the following new identities for $\boldsymbol\ell=\boldsymbol\ell^{\textrm{r}}+i\boldsymbol\ell^{\textrm{i}}$ and $\tilde{\vct{a}}=\tilde{\vct{a}}^{\textrm{r}}+i\tilde{\vct{a}}^{\textrm{i}}$, r real part, i imaginary part, assuming non-zero denominators:

$$\frac{\ell_{3}^{\textrm{i}}}{\tilde{a}_{2}^{\textrm{r}}}=-\frac{\tilde{a}_{3}^{\textrm{r}}}{\ell_{2}^{\textrm{i}}}=\frac{\tilde{a}_{3}^{\textrm{i}}}{\ell_{2}^{\textrm{r}}}=e, \spc \frac{\ell_{1}^{\textrm{i}}}{\ell_{2}^{\textrm{r}}}=-\frac{\ell_{1}^{\textrm{r}}}{\ell_{2}^{\textrm{i}}}=\frac{\tilde{a}_{1}^{\textrm{i}}}{\tilde{a}_{2}^{\textrm{r}}}=-\sqrt{1-e^{2}}.$$

\npgni We also obtain, where $\psi_{\textrm{sc}}\ne0$,

$$\cos \theta=\frac{\lambda \ell_{3}^{\textrm{r}}+\tilde{a}_{1}^{\textrm{r}}\tilde{a}_{2}^{\textrm{i}}}{(\ell_{3}^{\textrm{r}})^{2}+(\tilde{a}_{1}^{\textrm{r}})^{2}},\spc \sin \theta=\frac{\lambda \tilde{a}_{1}^{\textrm{r}}-\ell_{3}^{\textrm{r}}\tilde{a}_{2}^{\textrm{i}}}{(\ell_{3}^{\textrm{r}})^{2}+(\tilde{a}_{1}^{\textrm{r}})^{2}};\spc \sin \theta=e,$$

\npgni which are generalisations of Newton’s results for planetary motion. (See Ref.[22]).

\npgni Our last word in this section is used to draw attention to some intriguing anti-gravity effects in our semi-classical theory best illustrated by the effective potential,

$$V_{\textrm{eff}}=V-|\boldsymbol\nabla R|^2=-\frac{\mu}{r}-|\boldsymbol\nabla R|^2,$$\vspace{-5mm}

\npgni in the Keplerian case. For our astronomical elliptic state, $\psi_{n,e}\sim\textrm{exp}\left(\dfrac{R+iS}{\hbar}\right)$ as $\hbar\sim0$, correct to the leading term in $\dfrac{1}{r}$ as $r\sim\infty$,

$$V_{\textrm{eff}}\sim+\dfrac{\mu}{r}+O\left(\dfrac{1}{r^2}\right).$$\vspace{-8mm}

\npgni This is so, even though we can prove that for

$$R(\vct{X}_{t=0})>\lambda\left(2\left(\dfrac{1-e}{1+e}\right)-\ln4\right)$$

\npgni the particle orbits spiral onto the corresponding Keplerian ellipse. A striking example is the case $e=0$, in 2-dimensions,  where it is easy to compute

$$V_{\textrm{eff}}=+\frac{\mu}{r}-\frac{\lambda^2}{r^2}-\frac{\mu^2}{\lambda^2},\;\;\;r>0,$$\vspace{-8mm}

\npgni where $|\boldsymbol\nabla R+\boldsymbol\nabla S|\rightarrow0$ as $r=|\vct{r}|\nearrow\infty$. So semi-classical particles on the outer rim of the condensing cloud in the state $\psi_{n,e}$, if they are at rest in our inertial frame, should be repelled by this gravitating system!\vspace{-3mm}

\npgni It follows from the last equation that: $V'_{\textrm{eff}}(2a)=0$, $V'_{\textrm{eff}}(r)>0$ for $r<2a$ and $V'_{\textrm{eff}}(r)<0$ for $r>2a$ so that $r=2a$ is a local maximum of $V_{\textrm{eff}}(r)$. We also note that $V''_{\textrm{eff}}(3a)=0$. So from the energy conservation equation WIMPish particles describing circular spiral orbits will be slowing down when outside the circle, $r=2a$, and speeding up inside. This anti-gravity bump could 
provide an elementary test of our theory in simple protoplanetary nebulae such as we envisage formed Jupiter, or, as we argue, in circular spiral galaxies.\vspace{-3mm}

\npgni The last result also applies to the isotopic harmonic oscillator potential and the corresponding circular spiral orbits. In this case a simple calculation yields:\hspace{5mm}$V'_{\textrm{eff}}(2^{\frac{1}{4}}a)=0$, $a$ being the radius, so this result is relevant for celestial bodies such as Trojan asteroids being formed at Lagrange equilibrium points. Similar results still obtain for small orbital eccentricities, $e\sim0$. In the Keplerian case, correct to first order in $e$, the anti-gravity bump is on the ellipse with equation, $\dfrac{2a}{r}=1+\dfrac{7}{4}e\cos{\theta}$ for the elliptical spiral corresponding to the classical orbit, $\dfrac{a}{r}=1+e\cos{\theta}$.\vspace{-2mm}

\npgni Needless to say the above result is true on average in Nelson's stochastic mechanics in the Bohr correspondence limit where more detailed results emerge. Here we content ourselves by quoting a final theorem, again relevant to the formation of Trojan asteroids. We will enlarge on this in our future work, taking into account the large deviations.

\begin{theorem}
The transition density for the 2-dimensional radial process corresponding to circular orbits of the quantum isotropic harmonic oscillator:

$$dr(t)=\left(\dfrac{(n-\frac{1}{2})\epsilon^2}{r(t)}-\omega r(t)\right)dt+\epsilon dB(t),\;\;\;r(t)>0,$$

\npgni for integer $n\ge1$ and where $B$ is a BM($\mathbb{R}$) process, is given by

\npgni \hspace{3mm} $p(x,s;y,t)=\dfrac{2\omega}{\epsilon^2}\dfrac{\textrm{e}^{(n-1)\omega(t-s)}}{(1-\textrm{e}^{-2\omega(t-s)})}\dfrac{y^n}{x^{n-1}}\exp{\Bigg\{\dfrac{-\omega(y^2+x^2\textrm{e}^{-2\omega(t-s)})}{\epsilon^2(1-\textrm{e}^{-2\omega(t-s)})}\Bigg\}}$

\npgni \hspace{90mm} $\times\; I_{n-1}\left(\dfrac{2\omega xy}{\epsilon^2(\textrm{e}^{\omega(t-s)}-\textrm{e}^{-\omega(t-s)})}\right)$,

\npgni where $t>s\ge0$, $r(s)=x$, $r(t)=y$ and $I$ is the modified Bessel function.

\end{theorem}

\npgni Using this result and applying Laplace's principle (see Ref.[8]) we see that in the Bohr correspondence limit, $n\nearrow\infty$ and $\epsilon\searrow0$ with $n\epsilon^2=\lambda$ the corresponding radial process transition density has an extremum when

$$y^2=\dfrac{\lambda}{\omega}+\left(x^2-\dfrac{\lambda}{\omega}\right)\textrm{e}^{-2\omega t},$$

\npgni where we have taken $s=0$ for simplicity. This last equation is precisely the semi-classical orbit, at time $t$, for $\epsilon=0$, which in turn tends to the classical circular orbit 
as $t\nearrow\infty$.

\npgni \textbf{Exercise}

\npgni If $y_{\textrm{sc}}(x,t)$ is the equation of the semi-classical orbit starting at $x$ at $t=0$, setting $y_{\textrm{sc}}^2=y^2$ for the above $y$, for initial WIMP particle density,

$$\mathbb{P}_{0}(W\in dx)=\rho_{0}(x)dx,$$

\npgni then in the Bohr correspondence limit,

$$\mathbb{P}_{t}(W\in dy)=\rho_{0}(x_{0}(y,t))\Big|\dfrac{\partial x_{0}}{\partial y}(y,t)\Big|dy,$$

\npgni where $x_{0}=x_{0}(y,t)$ is the unique solution of $y_{\textrm{sc}}(x_{0},t)=y$.

\npgni In the infinite time limit this gives the WIMP particle density on the classical orbit for the formation of e.g. the Trojan asteroids near the $\mathscr{L}_{4,5}$ Lagrange points when their orbits are circular. This result generalises for small eccentricities.

\subsection {On the (R,V) Equation, Vector Potentials and Fluid Models}

\npgni Solving the $(R,V)$ equation of (NQG I) amounts to resolving the Nelson problem of our $S$ term not being a gradient field. We have seen that this difficulty does not arise for potentials $V=-\dfrac{\mu}{r}$ or $V=Kr^2$ in 2-dimensions or 3-dimensions in the explicit case of the astronomical elliptic states, but for other potentials the problem remains. Obviously, when our WIMPish particles are charged the $(R,V)$ equation of (NQG I) can always be solved by including a vector potential, replacing $\boldsymbol{\nabla}S$ by $(\boldsymbol{\nabla}S-\vct{A})$, $\textrm{curl}\vct{A}\ne\vct{0}$, $\textrm{div}\vct{A}=0$. Moreover, when one considers the data from Hubble and the James Webb telescopes one is lead to conclude that a fluid model would be a more immediate way to elucidate what one sees in a neighbourhood of a typical ring corresponding to one of our astronomical elliptic states. So in this section we give such a model based on the semi-classical mechanics of a linear superposition of Schrödinger stationary state wave functions, $\psi\sim\textrm{exp}\left(\frac{R+iS}{\epsilon^2}\right)$ as $\epsilon\sim0$, including a vector potential $\vct{A}$ and scalar potential $V$. When $\vct{A}=\boldsymbol{0}$, our astronomical elliptic states provide a paradigm of fluids spiralling toward Keplerian ellipses obeying a version of Newton's laws forming a ring system. Our fluid model gives results for quite a large class of time-dependent states in the semi-classical limit and more generally because our fluid is posited to have both viscosity, $\sigma^2$, and vorticity; as you will see the bigger $\sigma^2$, the smaller the WIMP mass.

\npgni In our arguments we accept the primacy of the log particle density $R$ and the Hamilton Jacobi Function $S$ as required by the above asymptotics consistent with an entropy $\mathcal{E}=\dfrac{(R-R_{\textrm{max}})}{\epsilon^2}$ for a superposition of our astronomical elliptic eigenstates but here, as we shall see, the entropy is different. We assume we have a classical, planar, closed curve $C_{0}$, representing the classical orbit on 
which $R$ achieves its global maximum $R_{\textrm{max}}$, typically $C_{0}$ will be a Keplerian ellipse. The level surfaces, $R=c$, will then be approximately toroidal surfaces forming tubes centred on $C_{0}$. In a neighbourhood of $C_{0}$, $|\boldsymbol{\nabla}R|\sim0$, and fluid particle velocity, $\vct{v}=\boldsymbol{\nabla}R+\boldsymbol{\nabla}S-\vct{A}\sim\boldsymbol{\nabla}S-\vct{A}$, with $\textrm{curl}\vct{A}\not\equiv\boldsymbol{0}$, when the fluid is not irrotational. In a neighbourhood of $C_{0}$, since

$$\boldsymbol\nabla(2^{-1}\vct{v}^2)=(\vct{v}.\boldsymbol\nabla)\vct{v}+\vct{v}\times\textrm{curl}\vct{v}\sim|\vct{v}|\dfrac{\partial}{\partial S}\vct{v}-(\vct{A}.\boldsymbol\nabla)\vct{A}+\textrm{curl}\vct{A}\times\vct{v},$$

\npgni this would be a good place to look for the effects of $\vct{A}$ by concentrating on the vorticity term $\textrm{curl}\vct{A}\times\vct{v}$, which could reinforce dark matter effects. Always remember this neighbourhood is a collision zone and our WIMPish particles may have captured charge by combining with other particles and will therefore no longer be in an astronomical elliptic state and could be subject to vector potentials from magnetic fields. Nevertheless, we will show how semi-classical analysis can explain their behaviour in quite general circumstances.

\npgni In 3-dimensions, to ensure consistency with the Schrödinger equation, we imitate our 2-dimensional prescription of writing our putative solution, $\boldsymbol{\xi}=\boldsymbol{\nabla}S-\vct{A}$, of the $(R,V)$ equation as $\boldsymbol{\xi}=\lambda(\boldsymbol{\nabla}R)^{\perp}$, $\lambda=\sqrt{1+\dfrac{2(E-V)}{|\boldsymbol\nabla R|^2}}$, thereby ensuring energy conservation: $2^{-1}(|\boldsymbol{\xi}|^2-|\boldsymbol{\nabla}R|^2)+V=E$. To this end we concentrate on the level surfaces, $R=c$, with normal $\boldsymbol{\nabla}R$, at a point P on the level surface, $\overset{\rightharpoonup}{\textrm{OP}}=\vct{r}_{\textrm{p}}$, where the local normal coordinates $(u,v)$ are determined by the principal directions forming an orthogonal net on $R=c$.

\npgni The orthogonality condition $\boldsymbol{\xi}\perp\boldsymbol{\nabla}R$ means that $\xi_{w}=0$, where $w$ is the coordinate in the normal direction. Here $\boldsymbol{\xi}=\boldsymbol{\nabla}S-\vct{A}$, where possibly $\textrm{curl}\vct{A}\ne\boldsymbol{0}$. Anyway we can assume, redefining $S$ by a gauge transformation, that $\bigtriangleup S=\textrm{div}\boldsymbol{\xi}$, and $\textrm{div}\vct{A}=0$. What about $\textrm{curl}\vct{A}$? This is determined by cyclic permutations of the equations in $(u,v,w)$,

$$-\dfrac{1}{h_{u}h_{v}}\left(\dfrac{\partial}{\partial u}(h_{v}\xi_{v})-\dfrac{\partial}{\partial v}(h_{u}\xi_{u})\right)=(\textrm{curl}\vct{A})_{w}.$$

\npgni Cyclic permutations of the $(u,v,w)$ equations gives, if $\textrm{curl}\vct{A}=\boldsymbol{0}$,

$$\dfrac{1}{h_{v}h_{w}}\left(\dfrac{\partial}{\partial v}(h_{w}\xi_{w})-\dfrac{\partial}{\partial w}(h_{v}\xi_{v})\right)=0,$$

$$\dfrac{1}{h_{w}h_{u}}\left(\dfrac{\partial}{\partial w}(h_{u}\xi_{u})-\dfrac{\partial}{\partial u}(h_{w}\xi_{w})\right)=0,$$

$$\dfrac{1}{h_{u}h_{v}}\left(\dfrac{\partial}{\partial u}(h_{v}\xi_{v})-\dfrac{\partial}{\partial v}(h_{u}\xi_{u})\right)=0.$$

\npgni For $\xi_{w}=0$, these imply that $h_{v}\xi_{v}$ and $h_{u}\xi_{u}$ are functions of $(u,v)$ only and to solve the last equation we need a function $h=h(u,v)$, with

$$h_{v}\xi_{v}=\dfrac{\partial h}{\partial v},\;\;\;h_{u}\xi_{u}=\dfrac{\partial h}{\partial u},$$

\npgni $h$ being $C^2$. Let us now see how this determines the direction of $\boldsymbol{\xi}$ in $T_{P}$. We already know that $\oint\limits_{C}\boldsymbol{\xi}.d\vct{r}=0$ for any simple closed curve on $R=c$ from the existence of $h$, if it can be defined on the whole surface $R=c$. Now $|\boldsymbol{\xi}|^2=\land^2=\xi_{u}^2+\xi_{v}^2$, so $\xi_{v}=h_{v}^{-1}\dfrac{\partial h}{\partial v}$, $\xi_{u}=h_{u}^{-1}\dfrac{\partial h}{\partial u}$ and

$$\land^2=|\boldsymbol{\nabla}R|^2+2(E-V)=\left(h_{u}^{-1}\dfrac{\partial h}{\partial u}\right)^2+\left(h_{v}^{-1}\dfrac{\partial h}{\partial v}\right)^2,$$

\npgni which is not very difficult to satisfy. Let $\phi$ be the angle between $\boldsymbol{\xi}$ and $\vct{r}_{u}$ so that

$$\tan{\phi}=\dfrac{\xi_{v}}{\xi_{u}}=\dfrac{h_{u}\frac{\partial h}{\partial v}}{h_{v}\frac{\partial h}{\partial u}}.$$

\npgni So $\boldsymbol{\nabla}S=\boldsymbol{\xi}=\sqrt{|\boldsymbol{\nabla}R|^2+2(E-V)}(\cos{\phi}\;\hat{\vct{r}}_{u}+\sin{\phi}\;\hat{\vct{r}}_{v})$, $\xi_{w}=0$ and $\textrm{curl}\vct{A}=\boldsymbol{0}$.

\npgni To finally determine $\phi$ in our coordinate neighbourhood of P on the level surface, $R=c$, we need more data. We choose the angular momentum, $\vct{L}(\vct{r})=\vct{r}\times(\boldsymbol{\nabla}R+\boldsymbol{\nabla}S)$, which we assume is known at least approximately, where $\vct{r}=\overset{\rightharpoonup}{\textrm{OQ}}$, Q being in the neighbourhood of P being different for each elliptical orbit in the ring system. Then, if $|\vct{r}\times\boldsymbol{\nabla}S|=r|\boldsymbol{\nabla}S|\sin{\alpha}$,

$$\sin{\alpha}=\dfrac{|\vct{L}(\vct{r})-\vct{r}\times\boldsymbol{\nabla}R|}{r\sqrt{|\boldsymbol{\nabla}R|^2+2(E-V)}}.$$

\npgni Using local coordinates,

\npgni $\widehat{\boldsymbol{\nabla}R}=(0,0,1)$, $\widehat{\boldsymbol{\nabla}S}=(\cos{\phi},\sin{\phi},0)$, $\vct{r}=\overset{\rightharpoonup}{\textrm{OQ}}=(x,y,z)$ so that

$$\cos{\alpha}=\sqrt{(x^2+y^2)}\sin{(\phi+\chi)},\;\;\textrm{where}\;\;\chi=\arctan{\left(\frac{y}{x}\right)},$$

\npgni determining $\phi$. Needless to say for central forces, if P is in the neighbourhood of $C_{0}$, characterising a typical ring corresponding to $R_{\textrm{M}}=R_{j}$, $\vct{L}(\vct{r})\cong\vct{L}_{j}$, a constant.

\npgni Modulo finding the integral function $h(u,v)$ this solves $\textrm{curl}\vct{A}=\boldsymbol{0}$ locally in the tangent plane $\textrm{T}_{p}$ at a point P on the level surface $R=c$. Only if this local solution can be extended to the whole of space can we say with confidence $\textrm{curl}\vct{A}\equiv\boldsymbol{0}$. So it is likely we need to include vector potentials. Of course if we use our solution of the Schrödinger equation with the desired asymptotics all of the above must hold good. To recapitulate we need:-

$$\boldsymbol{\xi}.\boldsymbol{\nabla}R=0,\;\;2^{-1}(\boldsymbol{\xi}+\boldsymbol{\nabla}R)^2+V_{\textrm{eff}}(\vct{q})=E,\;\;V_{\textrm{eff}}=V-|\boldsymbol{\nabla}R|^2,\;\;\boldsymbol{\xi}=\boldsymbol{\nabla}S-\vct{A},$$

\npgni for the corresponding stationary state solution, $\psi\sim\textrm{exp}\left(\frac{R+iS}{\epsilon^2}\right)$ as $\epsilon\sim0$ of

$$H_{Q}(\vct{q},\vct{p})\psi=E\psi,\;\;\textrm{where}\;\;H(\vct{q},\vct{p})=2^{-1}\left(\vct{p}-\vct{A}(\vct{q})\right)^2+V_{\textrm{eff}}(\vct{q}),$$

\npgni with constraint $(\vct{p}-\boldsymbol{\nabla}R).\boldsymbol{\nabla}R=0$, in the classical limiting case. Here we presuppose we have solutions of Hamilton's equations:

$$\dot{\vct{q}}=\dfrac{\partial H}{\partial\vct{p}},\;\;\;\;\;\;\dot{\vct{p}}=-\dfrac{\partial H}{\partial\vct{q}},$$

\npgni given by $\vct{p}=\boldsymbol{\nabla}S(\vct{q})+\boldsymbol{\nabla}R(\vct{q})-\vct{A}$ and $\dot{\vct{q}}=\vct{p}-\vct{A}=\boldsymbol{\nabla}S(\vct{q})-\vct{A}+\boldsymbol{\nabla}R(\vct{q})$. This is in line with taking the Bohr correspondence limit of Nelson's stochastic mechanics for scalar and vector potentials, although it is the first time we have spelled this out in detail. (Here we assume $V$, $\vct{A}$ and $E$ are per unit mass and we still retain at this stage $(\boldsymbol{\nabla}S-\vct{A}).\boldsymbol{\nabla}R=0$ as above).

\npgni So what you may say is the form of the Hamilton equations? Setting particle velocity\vspace{-5mm}

\npgni $\vct{v}=\dot{\vct{q}}=(\boldsymbol{\nabla}S-\vct{A}+\boldsymbol{\nabla}R)$, $(\boldsymbol{\nabla}S-\vct{A}).\boldsymbol{\nabla}R=0$, and taking the gradient of the energy equation gives,

$$\boldsymbol{\nabla}_{q}(2^{-1}\vct{v}^2+V_{\textrm{eff}}(\vct{q}))=
(\vct{v}.\boldsymbol{\nabla}_{q})\vct{v}+\vct{v}\times\textrm{curl}\vct{v}+\boldsymbol{\nabla}_{q}V_{\textrm{eff}},$$

\npgni which is reminiscent of fluid dynamics. In fact, setting the convected derivative,

$$\dfrac{D}{\partial t}=\dfrac{\partial}{\partial t}+\vct{v}.\boldsymbol{\nabla},$$

\npgni where the first term only contributes in time dependent cases we have a Burgers-Zeldovich equation with vorticity, viz

$$\dfrac{D\vct{v}}{\partial t}=-\boldsymbol{\nabla}V_{\textrm{eff}}+\textrm{curl}\vct{v}\times\vct{v},$$

\npgni where $\vct{v}=(\boldsymbol{\nabla}S-\vct{A}+\boldsymbol{\nabla}R)$, with constraint $(\boldsymbol{\nabla}S-\vct{A}).\boldsymbol{\nabla}R=0$ for our classical limit. To solve the latter problem it is tempting to set $|\boldsymbol{\nabla}R|\sim0$, a slowly varying $R$, so as to consider the Hamilton-Jacobi equation in the time-dependent case,

$$\dfrac{\partial S_{t}}{\partial t}+H(\vct{q},\boldsymbol{\nabla}S_{t}(\vct{q}))=0,\;\;\;S_{t=0}=S_{0}.$$

\npgni As we shall see, this inevitably leads to a Euclidean version of our theory, i.e. a Schrödinger-Heat equation which is not obviously relevant to the formation of planets, stars and galaxies, but see the example below. We shall return to the time-dependent problem after considering the example.

\npgni \textbf{Exercise}

\npgni Prove that the 2-dimensional $(R,V)$ problem in (NQG I) can always be solved by including a vector potential $\vct{A}$, satisfying $A_{z}=0$ and $-\bigtriangleup\vct{A}=\boldsymbol{\nabla}^{\perp}(\textrm{div}(\lambda\boldsymbol{\nabla}R))$.

\npgni \textbf{Example}

\npgni Let $\Psi(\vct{r},\epsilon^2)$, $\vct{r}\in\mathbb{R}^3$, be the circular stationary state solution of the Schrödinger equation of a unit particle moving in the Coulomb potential, $-\dfrac{\mu}{r}$, where $r=|\vct{r}|$, i.e. for energy $E$,

$$-\dfrac{1}{2}\epsilon^4\bigtriangleup\Psi-\dfrac{\mu}{r}\Psi=E\Psi.$$

\npgni Writing $\epsilon^2=i\sigma^2$ and $U(\vct{r},\sigma^2)=\Psi(\vct{r},i\sigma^2)$, formally at least, $U$ is a solution of the Schrödinger-Heat equation,

$$\dfrac{1}{2}\sigma^4\bigtriangleup U-\dfrac{\mu}{r}U=EU.$$

\npgni If $U=\textrm{exp}\left(\dfrac{S-iR}{\sigma^2}\right)$, for real $R$ and $S$, is a complex-valued solution of the above equation then $U=\textrm{exp}\left(\dfrac{S+R}
{\sigma^2}\right)$ is a real solution of the modified heat equation,\vspace{3mm}

$$\dfrac{1}{2}\sigma^4\bigtriangleup U+\left(-\dfrac{\mu}{r}-|\boldsymbol{\nabla}R|^2\right)U=EU.$$

\npgni Transforming the appropriate stationary state solution, $\Psi$, (see ref.[12]), we can construct an exact solution $U$ with

$$R=-\dfrac{\mu}{\lambda}r+\dfrac{\lambda}{2}\ln{(x^2+y^2)}+\sigma^2\tan^{-1}\left(\dfrac{y}{x}\right),$$

$$S=\lambda\tan^{-1}\left(\dfrac{y}{x}\right)-\dfrac{\sigma^2}{2}\ln{(x^2+y^2)},$$

\npgni where $\lambda>0$ is defined by $E=-\dfrac{\mu^2}{2\lambda^2}$. We also note that $\sigma$ need not be small and $\boldsymbol{\nabla}R.\boldsymbol{\nabla}S\ne0$.

\npgni If $\vct{v}=\boldsymbol{\nabla}R+\boldsymbol{\nabla}S$ defines the deterministic part of the particle velocity then

$$\dfrac{d\vct{v}}{dt}=-\boldsymbol{\nabla}V_{\textrm{eff}},\;\;\textrm{where}\;\;V_{\textrm{eff}}=\dfrac{\mu}{\lambda r}(\lambda-\sigma^2)-\dfrac{\lambda^2+\sigma^4}{x^2+y^2}-\dfrac{\mu^2}{\lambda^2}.$$

\npgni We now assume that $0<\sigma^2<\lambda$. This system can be solved exactly, showing that the particle paths spiral on to the circular orbit with radius $r=\dfrac{\lambda(\lambda-\sigma^2)}{\mu}$ in the plane $z=0$. Moreover the $3^{\textrm{rd}}$ component of angular momentum is equal to $(\lambda+\sigma^2)$.

\npgni This example leads us to ask, could this be a model for the formation of spiral galaxies such as the Whirlpool galaxy, Messier 51? In addition to the spiral nature of the solution if $\sigma^2\sim\lambda$ the $3^{\textrm{rd}}$ component of angular momentum is approximately \textbf{twice} the classically predicted value! Could this help to explain dark matter data reproducing the observed rotation curve for galaxies' gaseous parts.

\npgni Other applications of this technique include cyclone cloud formation, Phytoplankton ocean swirls, bubble chamber photographs showing spiral paths and atmospheric rocket fuel spirals.

\npgni The following theorem encapsulates the main ideas underlying the last example and provides a source of solutions of Burgers-Zeldovich equations with vorticity, which can be superposed making them appropriate to galactic evolution and dark matter data. Firstly, recall the Schrödinger equation for a time-dependent wave function, $\Psi_{t}$, for a unit mass particle subject to a scalar potential $V$ and vector potential $\vct{a}$,

$$\dfrac{\partial\Psi_{t}}{\partial t}=\left(\dfrac{\sigma^2}{2}\bigtriangleup+\;\vct{a}.\boldsymbol{\nabla}+\dfrac{1}{\sigma^2}\left(V+\dfrac{\vct{a}^2}
{2}\right)\right)\Psi_{t},\;\;\;\textrm{div}\vct{a}=0,\;\;\;\dfrac{\partial\vct{a}}{\partial t}=0,$$

\npgni in this simple example, with $\sigma^2=i\hbar$. By the corresponding Schrödinger-Heat equation we mean, modulo above assumptions,

$$\dfrac{\partial U_{t}}{\partial t}=\left(\dfrac{\sigma^2}{2}\bigtriangleup+\;\vct{a}.\boldsymbol{\nabla}+\dfrac{1}{\sigma^2}\left(V+\dfrac{\vct{a}^2}{2}\right)\right) U_{t},\;\;\sigma^2\in\mathbb{R}_{+}.\;\;\;\;\;(\textrm{SHE})$$

\npgni \textbf{The Hopf-Cole transformation} for $\mathcal{S}_{t}$, $U_{t}=\textrm{exp}\left(-\dfrac{\mathcal{S}_{t}}{\sigma^2}\right)$ reduces this to:

$$\dfrac{\partial\mathcal{S}_{t}}{\partial t}+\dfrac{1}{2}|\boldsymbol{\nabla}\mathcal{S}_{t}-\vct{a}|^2+V=\dfrac{\sigma^2}{2}\bigtriangleup\mathcal{S}_{t},\;\;\mathcal{S}_{t}=-\sigma^2\ln{U_{t}}.$$

\npgni All of the above need to be solved for the given initial conditions at $t=0$.

\begin{theorem}

Corresponding to the complex-valued solution, $U_{t}=\textrm{exp}\left(\dfrac{iR_{t}-S_{t}}{\sigma^2}\right)$ of the SHE, we have two real valued solutions of a modified SHE with $V\rightarrow V-|\boldsymbol{\nabla}R_{t}|^2$,

$$U_{t}=\textrm{exp}\left(-\dfrac{\mathcal{S}_{t}}{\sigma^2}\right),$$

\npgni where $\mathcal{S}_{t}=(\pm R_{t}\pm S_{t})$. Here we consider the +ve signs.

\end{theorem}

\begin{proof}
Dropping the subscript $t$, we know that:-

$$\dfrac{\partial R}{\partial t}=\dfrac{\sigma^2}{2}\bigtriangleup R-\boldsymbol{\nabla}R.\boldsymbol{\nabla}S+\vct{a}.\boldsymbol{\nabla}R,\;\;\;\;\;\;(i)$$

$$-\dfrac{\partial S}{\partial t}=-\dfrac{\sigma^2}{2}\bigtriangleup S+\dfrac{1}{2}(|\boldsymbol{\nabla}S|^2-|\boldsymbol{\nabla}R|^2)-\vct{a}.\boldsymbol{\nabla}S+V+\dfrac{\vct{a}^2}{2},\;\;\;\;(ii)$$

\npgni Subtracting gives

$$\dfrac{\partial\mathcal{S}}{\partial t}+\dfrac{1}{2}|\boldsymbol{\nabla}\mathcal{S}-\vct{a}|^2+V-|\boldsymbol{\nabla}R|^2=\dfrac{\sigma^2}{2}\bigtriangleup\mathcal{S},\;\;\;\mathcal{S}=(R+S).\;\;\;\;\;(iii)$$
 
\end{proof}

\npgni To recover our semi-classical eigenfunction results, where $\vct{a}\equiv\vct{0}$, we need $\dfrac{\partial R}{\partial t}=0$ and $\dfrac{\partial S}{\partial t}=-E$, for astronomical elliptic states. What happens more generally in the time-dependent case? Say for the sake of argument, $R_{t}=\textrm{O}(\sigma^2)$, in the SHE case $R_{t=0}(\vct{x})=\sigma^2\ln{T_{0}}(\vct{x})$, so that

$$U_{t=0}(\vct{x})=T_{0}(\vct{x})\textrm{exp}\left(-\dfrac{S_{0}(\vct{x})}{\sigma^2}\right),$$

\npgni $T_{0}>0$ and $\int T_{0}^2(\vct{x})d\vct{x}<\infty$, for a finite WIMPish fluid mass. As we shall see the elementary formula of Elworthy-Truman gives for fluid density $\rho_{0}$,

$$\rho_{0}^{1/2}(\vct{x},t)=\lim_{\sigma^2\rightarrow0}\textrm{exp}\left(\dfrac{S_{t}(\vct{x})}{\sigma^2}\right)U_{t}(\vct{x})=T_{0}(\vct{x}_{0},t)|J|^{1/2},\;\;\;J=\Big|\dfrac{\partial\vct{x}_{0}}{\partial\vct{x}}(\vct{x},t)\Big|,$$

\npgni (the Jacobian determinant) and where $S_{t}$ is the Hamilton-Jacobi function satisfying:

$$\dfrac{\partial S_{t}}{\partial t}+\dfrac{1}{2}|\boldsymbol{\nabla}S_{t}-\vct{a}|^2+V=0,\;\;\textrm{with}\;\;S_{t=0}=S_{0}.$$

\npgni This result is only true for sufficiently small time $t<T$, where $T$ is the caustic time, $\vct{x}_{0}(\vct{x},t)$ being the solution of $\phi_{t}(\vct{x}_{0})=\vct{x}$ for the classical flow map, $\phi_{s}$, defined by:

$$\phi_{s}(\vct{x}_{0})=\vct{X}(s,\vct{x}_{0},\boldsymbol{\nabla}S_{0}(\vct{x}_{0}),\;\;\;s\in(0,t),\;\;\;t<T,$$

\npgni $\vct{X}(s)=\vct{X}(s,\vct{x}_{0},\boldsymbol{\nabla}S_{0}(\vct{x}_{0})$, satisfying, $\vct{X}(0)=\vct{x}_{0}$, $\dot{\vct{X}}(0)=\boldsymbol{\nabla}S_{0}(\vct{x}_{0})$ and

$$\ddot{\vct{X}}(s)=-\boldsymbol{\nabla}V(\vct{X}(s))+\dot{\vct{X}}(s)\times\textrm{curl}(\vct{a}(\vct{X}(s))),\;\;\;s\in(0,t),\;\;\;t<T,$$

\npgni the classical equation corresponding to the Hamilton-Jacobi function and $T$ being the time up to which $\phi_{s}$ is a diffeomorphism, $\phi_{s}:\mathbb{R}^d\rightarrow\mathbb{R}^d$, this occurring when infinitely\vspace{-4mm}

\npgni many classical paths focus at a point and $\Big|\dfrac{\partial\vct{x}_{0}}{\partial\vct{x}}\Big|$ blows up or is zero.

\npgni It turns out that $\vct{v}_{t}=(\boldsymbol{\nabla}S_{t}-\vct{a})$ is the fluid velocity and as you would expect for $\rho_{0}(\vct{x},t)=T_{0}^2(\vct{x}_{0}(\vct{x},t)|J|$,

$$\dfrac{\partial\rho_{0}}{\partial t}+\textrm{div}(\rho_{0}\vct{v}_{t})=0$$

\npgni i.e. $\rho_{0}$ is the fluid density and for the fluid velocity field, $\vct{v}_{t}=\boldsymbol{\nabla}\mathcal{S}_{t}(\vct{x})-\vct{a}(\vct{x})$, in the limit as $\sigma^2\rightarrow0$,

$$\dfrac{\partial\vct{v}_{t}}{\partial t}+(\vct{v}_{t}.\boldsymbol{\nabla})\vct{v}_{t}+\vct{v}_{t}\times\textrm{curl}\vct{v}_{t}=-\boldsymbol{\nabla}V.$$

\npgni So we have a Burgers-Zeldovich fluid with vorticity. If one prefers one can, of course, retain the terms in $\sigma^2$ and consider

$$\dfrac{\partial S_{t}}{\partial t}+\dfrac{1}{2}|\boldsymbol{\nabla}\mathcal{S}_{t}-\vct{a}|^2+V=\dfrac{\sigma^2}{2}\bigtriangleup\mathcal{S}_{t}$$

\npgni and obtain solutions of Burgers-Zeldovich fluids with voticity and viscosity as in our example.

\npgni All this is possible in the framework of the 'Elworthy-Truman elementary formula' which we quote in its simplest version below. (For generalisations see Ref.[40] which is dedicated to David Elworthy).

\npgni Recall $\mathcal{S}_{t}=S(\vct{x},t)$ in the limit as $\sigma^2\rightarrow0$, the solution of our Hamilton-Jacobi equation is given by,

$$S_{t}(\vct{x})=S_{0}(\vct{x}_{0}(\vct{x},t))+\int_{0}^{t}\mathcal{L}(\vct{X}(s),\dot{\vct{X}}(s))ds,$$

\npgni for Lagrangian, $\mathcal{L}$, when $\dfrac{\partial\vct{a}}{\partial t}=\vct{0}$, $\textrm{div}\vct{a}=0$, for simplicity,

$$\mathcal{L}(\vct{X}(s),\dot{\vct{X}}(s))=2^{-1}(\dot{\vct{X}}(s))^2+\dot{\vct{X}}(s).\vct{a}(\vct{X}(s))-V(\vct{X}(s)),$$

\npgni where $\vct{X}(s)=\vct{X}(s,\vct{x}_{0},\boldsymbol{\nabla}S_{0}(\vct{x}_{0})\rceil_{\vct{x}_{0}=\vct{x}_{0}(\vct{x},t)}$.

\npgni We require the diffusion process, $\vct{Y}_{s}^\sigma$, to be non-explosive and satisfy

$$d\vct{Y}_{s}^{\sigma}=\sigma d\vct{B}(s)-(\boldsymbol{\nabla}S_{t-s}(\vct{Y}_{s}^{\sigma})-\vct{a}(\vct{Y}_{s}^{\sigma}))ds,\;\;\;\vct{Y}_{s}^{\sigma}\rceil_{s=0}=\vct{x},$$

\npgni where $\vct{B}(s)$ is a $\textrm{BM}(\mathbb{R}^3)$ process, for $s\in(0,t)$, $t<T$, the caustic time where $J=\Big|\dfrac{\partial\vct{x}_{0}}{\partial\vct{x}}(\vct{x},t)\Big|$ blows up or is zero. Then in its simplest form the elementary formula reads:

\begin{theorem}

Modulo the above assumptions and mild boundedness of e.g. $V$, $T_{0}$ etc.

$$\textrm{exp}\left(\dfrac{\mathcal{S}_{t}(\vct{x})}{\sigma^2}\right)U_{t}(\vct{x})=\mathbb{E}\left(T_{0}(\vct{Y}_{t}^{\sigma})\textrm{exp}\left(-\frac{1}{2}\int_{0}^{t}\bigtriangleup S_{t-s}(\vct{Y}_{s}^{\sigma})ds\right)\right),$$

\npgni where $U_{0}(\vct{x})=T_{0}(\vct{x})\textrm{exp}\left(-\dfrac{\mathcal{S}_{0}(\vct{x})}{\sigma^2}\right)$.

\end{theorem}

\npgni Evidently $U_{t}(\vct{x})$ is given by 'a sum over paths' formula, where the paths are essentially the sample paths of the Nelson diffusion process for the corresponding Schrödinger equation, as such they obey a $2^{\textrm{nd}}$ Law of Motion in the form of the Nelson-Newton Law,

$$\textrm{Force}=\textrm{Mass}\times\textrm{Acceleration},\;\;\textrm{if}\;\;\sigma=\sqrt{\dfrac{\hbar}{m}},\;\;m\;\;\textrm{the WIMPish mass},$$

\npgni (see Nelson Refs.[23],[24]). A formal asymptotic series expansion in powers of $\sigma^2$ of the above expression, first obtained in Ref.[38] and elaborated upon in Refs.[6] and [39], gives solutions to the Burgers-Zeldovich fluid equations with viscosity and vorticity. Needless to say the non-linear Burgers-Zeldovich equations will inherit a linear superposition principle from the Schrödinger-Heat equation with entropy $\left(-\dfrac{\mathcal{S}_{t}}{\sigma^2}\right)$ so a minimal action principle emerges in this framework as $\sigma^2\rightarrow 0$. Compare this with the case of eigenstates where the entropy is $\mathcal{E}=\dfrac{(R-R_{\textrm{max}})}{\epsilon^2}$.  Obviously the formal asymptotic expansion can be made to apply to the Schrödinger equation itself with the identity, $\textrm{SE}\rightarrow\textrm{SHE}$ epitomised by

$$\textrm{exp}\left(\dfrac{R+iS}{\epsilon^2}\right)\rightarrow\left(\dfrac{iR-S}{\sigma^2}\right),\;\;\textrm{when}\;\;R=\textrm{O}(\sigma^2),$$

\npgni using the Elworthy-Truman elementary formula and joint results with Zhao and Davies.

\section{Conclusion}\vspace{-5mm}

\npgni In this paper we have presented three new results:- two new constants of the motion for the linearised restricted 3-body problem, an important isosceles triangle generalisation of Lagrange's equilateral triangle solution for the restricted case and explicit formulae for the quantum corrections to curvature and torsion for the trajectories of semi-classical particles subject to vector as well as scalar potentials. All these results are relevant to the problem of understanding the past and future histories of the Trojan asteroids. Our quantum mechanical results are relevant for WIMPish particles where the main interaction with 
matter is through Newtonian gravitational attraction, the particles being relatively massive and non relativistic. We postulate that such particles were involved in the formation of the Trojans.

\npgni WIMPs themselves are very elusive; it seems and some authors have suggested that they may be hiding near the Lagrange points, $\mathscr{L}_{4,5}$, of 3-body problems such as the Sun, Jupiter and Trojan asteroid system studied in this work. Moreover, the two new constants of the motion should be observable in the motion of the Trojan asteroids near the $\mathscr{L}_{4,5}$ Lagrange points. This is very timely as the space-shot Lucy has recently been launched to photograph, close-up, some of the Trojans. This could result in data to confirm or rebutt our theoretical predictions. Similar remarks are pertinent to our isosceles triangle solution of the restricted 3-body problem where there may be better examples in other solar systems. For the latter isosceles triangle case we have included some relevant information from immediately available data and an exciting possible application to the Hildans. Other intriguing results are included in the Appendix which are of independent interest.

\npgni A final word about our predictions concerning the curvature and torsion of particles' trajectories in our astronomical state - these particle orbits are spirals converging to Keplerian ellipses as we have proved elsewhere (Refs.[9],[10],[11]). Could this be the explanation of why, as far as Hubble's results are concerned, 70\% of galaxies are spirals and only about 20\% are elliptical, the elliptical ones being older than the spiral ones? We believe so and have given an acid test of the validity of our ideas, if the astronomers can measure the curvature and torsion of particles in the spiral tails of galaxies. We believe there is no data currently available as far as this is concerned.

\npgni As we have proved in earlier works, we have found hidden constants in the semi-classical mechanics for our astronomical elliptical states by taking the Bohr correspondence limits of what are essentially Pauli's identities for the Hydrogen atom. When seen in an astronomical context they lead to complex identities for the orbits of WIMPish particles. Yet again they could be observable in the tails of galaxies as they evolve - another effect of Newtonian quantum gravity.

\npgni As far us extending our results to more general potentials and time inhomogeneous systems for general states and fluid models opens new windows. The interested reader is referred to Ref.[37] for more results.

\section*{Acknowledgement and Dedication}\vspace{-8mm}

\npgni Our astronomical elliptic states were formed by taking the semi-classical limit of the atomic elliptic states of Lena, Delande and Gay (Ref.[17]) who proved these are the closest possible states to Keplerian elliptic states. Moreover, the cognoscenti will realise our indebtedness to Ed Nelson (Refs.[23],[24],[25]). Further we would like to dedicate this paper to our late friend, Robin Hudson, who was a shining light to us all. Finally we need to add that without the inspirational work in mathematical physics of Barry Simon and Michael Berry and the teaching and influence of David Williams this paper would not have been written. 

\section *{Appendix (Available Data and Keplerian $4^{\textrm{th}}$ Law)}\vspace{-5mm}

\npgni Here we push our isosceles triangle result to the absolute limit, assuming all relevant orbits are approximately circular. The result is a Keplerian $4^{\textrm{th}}$ Law for 3-body systems. The equation of motion in the centre of mass rest frame of the primary 2-body problem of the third body (of negligible mass) moving in the isosceles triangle configuration under the gravitational attraction of the two primaries, each moving in a circular orbit, is\vspace{-5mm}

$$\ddot{\boldsymbol{\rho}}=-\dfrac{\mu_{1}+\mu_{2}}{\left(\rho^2+\dfrac{\mu_{2}}{\mu_{1}}r_{2}^2\right)^{\frac{3}{2}}}\boldsymbol{\rho},$$\vspace{-5mm}

\npgni where $\boldsymbol{\rho}$ is the displacement of the third body from the centre of mass, $O$, of the two primaries. Here $\mu_{1}$ and $\mu_{2}$ are the gravitational 
masses of the primaries and $r_{2}$ is the distance of the second (smaller) primary from $O$. The only underlying assumption is that the angular momentum of $\textrm{P}_{3}$ (particle of negligible mass) about $O$, the mass centre of $\textrm{P}_{1}$ and $\textrm{P}_{2}$, is conserved, which only requires that $\bigtriangleup\textrm{P}_{1}\textrm{P}_{2}\textrm{P}_{3}$, in the rest frame, be isosceles. This holds even if the particle masses $\mu_{1}$ and $\mu_{2}$ vary in time and if the side lengths of $\bigtriangleup\textrm{P}_{1}\textrm{P}_{2}\textrm{P}_{3}$ change as long as the triangle remains isosceles. So the principle applies over the history of 3-body problems in very general circumstances. For simplicity though we here only considered nearly circular orbits to test our ideas, circles being centred at $O$. This gives a new approximate picture of the formation of the solar system and other 3-body systems, taking into account the gravitational effect of a massive planet such as Jupiter in our own solar system.

\npgni The above equation supports circular orbits connecting orbital radii, $\rho$, with orbital periods, $T_{3}$, of the third body giving rise to a Keplerian type $4^{\textrm{th}}$ law of motion:

$$\rho=\bigg\{\left(\dfrac{T_{3}}{T_{2}}\right)^{\frac{4}{3}}\left(1+\frac{M_{2}}{M_{1}}\right)^2-\frac{M_{2}}{M_{1}}\bigg\}^{\frac{1}{2}}r_{2},$$\vspace{-5mm}

\npgni where $M_{1}$ and $M_{2}$ ($M_{1}>M_{2}$) are the masses of the primaries and $T_{2}$ is the orbital period of the second primary. This result is valid for $\rho=r_{3}>\dfrac{r_{2}-r_{1}}{2}$, where $r_{1}$ is the distance of the larger primary from $O$.

\npgni Our solar system as well as the moons orbiting the planets provide the perfect test bed of our ideas about isosceles triangle configurations, conservation of angular momentum and the above formula. If we idealise the solar system as circular 3-body systems with the Sun and Jupiter as the primaries and the planetary moon systems as circular 3-body systems using the planet and its largest moon as the primaries, available data allows us to compare our predictions for the orbital radius of a given planet/moon orbiting with a given orbital period. The tables below show the results for the solar system and the moons of Jupiter in prograde motion.

\npgni \textbf{The Solar System}

\npgni In this case the orbital elements of the primaries are $M_{1}=1.989\times10^{30}\;\textrm{kg}$, the mass of the Sun, $M_{2}=1.898\times10^{27}\;\textrm{kg}$, the mass of Jupiter, $T_{2}=4331\;\textrm{days}$, the orbital period of Jupiter and $r_{2}=5.2\;\textrm{Au}$, the distance of Jupiter from the Sun.

\begin{center}
\begin{tabular}{|c||c|c|}
\hline
 \multicolumn{3}{|c|}{The Solar System} \\
 \hline
 Planet & Semi-Major Axis of Orbit (Au) & Predicted Radius,
 $\rho$ (Au) \\
 \hline
 Saturn & 9.57 & 9.54\\
 Uranus & 19.17 & 19.21\\
 Neptune & 30.18 & 30.07\\
 Pluto & 39.48 & 39.41\\
 \hline
\end{tabular} 
\end{center}

\begin{center}
Table 1 - The Solar System
\end{center}

\npgni These results show very good agreement for orbits satisfying $\rho=r_{3}>\dfrac{r_{2}-r_{1}}{2}$. It should also be noted that this formula could apply, not only to the Trojan and Greek asteroids, but also to the Hildan asteroids and to many of the asteroids in the asteroid belt again with good agreement.

\npgni \textbf{Moons of Jupiter}

\npgni In this case the idealised 3-body system comprises Jupiter, Ganymede (the largest moon) and a second moon. The orbital elements of the primaries are $M_{1}=1.898\times10^{27}\;\textrm{kg}$, the mass of Jupiter, $M_{2}=1.4819\times10^{23}\;\textrm{kg}$, the mass of Ganymede, $T_{2}=7.1545\;\textrm{days}$, the orbital period of Ganymede and $r_{2}=1070400\;\textrm{km}$, the semi-major axis of the orbit of Ganymede about Jupiter. The table below compares the semi-major axis with the predicted radius, $\rho$, for 12 moons in prograde motion with Ganymede.

\begin{center}
\begin{tabular}{|c||c|c|}
\hline
 \multicolumn{3}{|c|}{Moons of Jupiter} \\
 \hline
 Moon & Semi-Major Axis of Orbit (km) & Predicted Radius,
 $\rho$ (km) \\
 \hline
 Europa & 671100 & 667707\\
 Calisto & 1882700 & 1882701\\
 Themisto & 7507000 & 7399085\\
 Leda & 11170000 & 11162027\\
 Ersa & 11401000 & 11416931\\
 Himalia & 11460000 & 11457512\\
 Pandia & 11481000 & 11498630\\
 Lysithea & 11700800 & 11719412\\
 Elara & 11740000 & 11732671\\
 Dia & 12260300 & 12285637\\
 Carpo & 16990000 & 17086110\\
 Valetudo & 18694200 & 18823105\\
 \hline
\end{tabular} 
\end{center}

\begin{center}
Table 2 - Moons of Jupiter
\end{center}\vspace{-5mm}

\npgni Again we see good agreement between observation and theoretical prediction. Analysing moons around Mars, Saturn, Neptune and Uranus give similar agreement and also apply to moons in retrograde motion. However in all these cases the ratio $\dfrac{M_{2}}{M_{1}}$ is small and the formula for $\rho$ is approximately that given by Kepler's $3^{\textrm{rd}}$ law. This is not the case for Pluto and its largest moon Charon.

\npgni \textbf{The Moons of Pluto}\vspace{-3mm}

\npgni  Here our idealised 3-body system comprises Pluto, Charon and a second moon. The orbital elements of the primaries are $M_{1}=1303\times10^{19}\;\textrm{kg}$, the mass of Pluto, $M_{2}=158.6\times10^{19}\;\textrm{kg}$, the mass of Charon, $T_{2}=6.38723\;\textrm{days}$, the orbital period of Charon and $r_{2}=17536\;\textrm{km}$, the semi-major axis of the orbit of Charon about the centre of mass. The table below compares the semi-major axis with the predicted radius, $\rho$, for the other 4 moons of Pluto. We note that $\dfrac{M_{2}}{M_{1}}=0.12172$, Pluto and Charon are in 1:1 resonance and all the orbits are nearly circular. This system provides the perfect test for our ideas. The table below gives the results.

\begin{center}
\begin{tabular}{|c||c|c|}
\hline
 \multicolumn{3}{|c|}{The Moons of Pluto} \\
 \hline
 Moon & Semi-Major Axis of Orbit (km) & Predicted Radius,
 $\rho$ (km) \\
 \hline
 Styx & 42650 & 41881\\
 Nix & 48690 & 48272\\
 Kereros & 57780 & 57473\\
 Hydra & 64740 & 64522\\
 \hline
\end{tabular} 
\end{center}

\begin{center}
Table 3 - The Moons of Pluto
\end{center}\vspace{-3mm}

\npgni Another good agreement.

\npgni \textbf{Circumbinary Systems}\vspace{-4mm}

\npgni Perhaps even more remarkable is that our result also applies to several other known planetary systems which can be idealised as 3-body systems. In particular the Kepler-16, Kepler-34, Kepler-35 and Kepler-38 binary systems all have circumbinary planets orbiting the centre of the system. In addition they all have a significant value of $\dfrac{M_{2}}{M_{1}}$. Omitting the details our calculations give the following results comparing known data for the orbital semi-major (S-M) axes with the predicted radius values, $\rho$.

\begin{center}
\begin{tabular}{|c||c|c|c|c|}
\hline
 \multicolumn{5}{|c|}{Circumbinary Systems}\\
 \hline
 System & $\dfrac{M_{2}}{M_{1}}$ & Planet & S-M Axis of Orbit (Au) & $\rho$ (Au)\\
 \hline
 Kepler-16 & 0.29368 & Kepler-16b & 0.7048 & 0.6985\\
 Kepler-34 & 0.97414 & Kepler-34b & 1.0896 & 1.0836\\
 Kepler-35 & 0.91179 & Kepler-35b & 0.6035 & 0.5965\\
 Kepler-38 & 0.26238 & Kepler-38b & 0.4632 & 0.4604\\
 \hline
\end{tabular}
\end{center}

\begin{center}
Table 4 - Circumbinary Systems
\end{center}\vspace{-5mm}

\npgni This is just a snapshot of the planets that are known to exist, but we believe our results apply to many other systems which will exhibit this fundamental property. In particular, we believe the best place to look for isosceles triangle orbits for the restricted 3-body problem will be in amongst the Hildan and Trojan asteroids.

\npgni Data Sources:-

\npgni 1. NASA Planetary Factsheets: https://nssdc.gsfc.nasa.gov/planetary/factsheet/

\npgni 2. JPL Planetary Satellite Mean Elements: https://ssd.jpl.nasa.gov/sats/elem

\npgni 3. NASA Exoplanet Archive: https://exoplanetarchive.ipac.caltech.edu

\end{document}